\makeatletter
\let\@saved@DeclareHookRule\DeclareHookRule
\def\DeclareHookRule#1#2#3#4{%
  \def\@tempa{#1}\def\@tempb{begindocument}%
  \ifx\@tempa\@tempb\else
    \@saved@DeclareHookRule{#1}{#2}{#3}{#4}%
  \fi}
\makeatother
\documentclass[a4paper,twocolumn,11pt,unpublished]{quantumarticle}
\pdfoutput=1
\usepackage[english]{babel}
\usepackage[utf8]{inputenc}
\usepackage[T1]{fontenc}
\usepackage{mathtools}
\usepackage{amssymb}
\usepackage{amsthm}
\usepackage{bm}
\usepackage{braket}
\usepackage{graphicx}
\usepackage[numbers]{natbib}
\usepackage{algorithmic}
\usepackage{booktabs}
\usepackage{enumitem}
\usepackage{subcaption}
\usepackage{multirow}
\usepackage{array}

\newtheorem{theorem}{Theorem}[section]
\newtheorem{lemma}[theorem]{Lemma}
\newtheorem{corollary}[theorem]{Corollary}

\newtheorem{definition}[theorem]{Definition}
\theoremstyle{remark}
\newtheorem{remark}[theorem]{Remark}

\DeclareMathOperator{\Tr}{Tr}

\DeclareMathOperator{\Var}{Var}
\newcommand{\Fid}{\mathcal{F}}
\newcommand{\Dmax}{D_{\infty}}
\newcommand{\R}{\mathbb{R}}
\newcommand{\C}{\mathbb{C}}
\newcommand{\calH}{\mathcal{H}}

\title{Optimal Quantum Differential Privacy via Fisher Information Spectral Analysis}

\author{Justice Owusu Agyemang}
\author{Jerry John Kponyo}
\author{Elliot Amponsah}
\author{Godfred Manu Addo Boakye}
\affiliation{Quantum and Assistive Technologies Lab, Kwame Nkrumah University of Science and Technology, Kumasi, Ghana}

\begin{document}

\maketitle

\begin{abstract}
The Quantum Fisher Information (QFI) metric governs a fundamental duality: it quantifies both how precisely a parameter can be estimated (metrology) and how distinguishable two quantum states are (privacy). We exploit this duality to establish a geometry-aware framework for quantum differential privacy (DP) that replaces isotropic depolarizing noise with direction-dependent noise aligned to the QFI eigenstructure of the quantum embedding. We prove six principal theorems: (1) the minimax-optimal mechanism concentrates the noise budget in the dominant QFI eigenmode, achieving $\varepsilon = (\Delta^2/2)\lambda_{\max}(1-c\gamma)$ with $O(d/\lambda_{\max})$ advantage; (2) mixed-state QFI decomposition reveals that dephasing in the adversary's basis \emph{increases} accessible information, while misaligned-basis dephasing provides constructive privacy amplification from hardware noise; (3) a tight privacy--utility uncertainty relation $\varepsilon \cdot (1-\Fid) \geq (\Delta^2/2)\Tr(F)/d$; (4) adaptive QFI estimation converging at $O(1/\sqrt{n})$ yields $1.92\times$ tighter bounds; (5) QFI-aligned composition saturates at $O(1)$ versus $O(k)$ for standard composition; and (6) hardware noise can be harnessed for privacy amplification. Adversarial vulnerabilities, Wasserstein guarantees, subspace projection, and a zero-knowledge audit protocol follow as corollaries. Results are validated on Qiskit Aer GPU simulations, IBM Quantum hardware (ibm\_fez, 156 qubits), and against classical DP baselines, achieving equivalent utility at $\varepsilon \approx 0.001$ versus $\varepsilon \approx 4800$ for classical DP.
\end{abstract}

\section{Introduction}

The convergence of quantum computing and machine learning has produced a rapidly advancing field encompassing quantum neural networks~\cite{cerezo2022,cong2019,havlicek2019}, quantum kernel methods~\cite{liu2021,havlicek2019,huang2021}, variational quantum algorithms~\cite{cerezo2021,bharti2022}, and quantum generative models~\cite{romero2017,dallaire2018,lloyd2018}. These algorithms exploit quantum superposition, entanglement, and interference to achieve potential exponential speedups in classification, regression, clustering, and generative modeling tasks over their classical counterparts~\cite{harrow2009,preskill2018}. As these algorithms transition from theoretical constructs to cloud-accessible services on platforms such as IBM Quantum, Amazon Braket, Google Quantum AI, and IonQ, the privacy of training data and model parameters emerges as a critical concern~\cite{watkins2023,du2022,senekane2017}.

Differential privacy (DP)~\cite{dwork2006,dwork2014} has become the gold standard for protecting individual privacy in classical machine learning, providing a mathematically rigorous framework that bounds the maximum information an adversary can extract about any individual training sample from the model's outputs. Its widespread adoption in industry~\cite{abadi2016,erlingsson2014,mcmahan2017} and government applications~\cite{dwork2010} attests to its practical effectiveness. The extension of DP to quantum computation was pioneered by Aaronson and Rothblum~\cite{aaronson2019}, who introduced gentle measurements as a quantum analogue of the Laplace mechanism. Zhou and Ying~\cite{zhou2017} formalized quantum DP in the framework of quantum programming languages and established the depolarizing channel $\Phi_\gamma(\rho) = (1-\gamma)\rho + \gamma I/d$ as the quantum equivalent of additive Gaussian noise, providing $(\varepsilon,0)$-DP with $\varepsilon = \ln(1 + d(1-F_{\min})/(\gamma(1-\gamma)))$, where $d$ is the Hilbert space dimension and $F_{\min}$ the minimum pairwise fidelity of the encoded states.

However, isotropic depolarizing noise is fundamentally suboptimal for quantum machine learning. Quantum embeddings map classical data $x \in \R^p$ onto a $d$-dimensional Hilbert space via $x \mapsto |\psi(x)\rangle$, where $d = 2^n$ for $n$ qubits. The resulting manifold of quantum states---the complex projective space $\mathbb{C}P^{d-1}$---carries a natural Riemannian metric: the Quantum Fisher Information (QFI)~\cite{braunstein1994,petz1996,helstrom1976,holevo2011}. This metric quantifies the \emph{anisotropic} distinguishability of nearby quantum states: the QFI eigenvalues $\{\lambda_k\}_{k=1}^p$ measure how much information about each input feature direction leaks into the quantum state. Directions with large $\lambda_k$ produce highly distinguishable output states (high privacy risk), while directions with small $\lambda_k$ produce nearly indistinguishable states (inherent privacy).

The QFI metric has been extensively exploited in quantum metrology for optimal parameter estimation~\cite{giovannetti2006,demkowicz2020,liu2020} and in quantum machine learning for natural gradient optimization~\cite{stokes2020,meyer2021}. Recent work~\cite{scala2025} has further shown how quantum noise reshapes the QFI spectrum through noise-induced equalization, redistributing sensitivity across parameter-space directions. Yet the profound implications of QFI geometry for differential privacy remain largely unexplored. The central insight of this paper is that \textbf{the QFI metric is precisely the privacy sensitivity metric for quantum embeddings}. If the isotropy of depolarizing noise is analogous to adding spherical Gaussian noise in classical DP, our geometry-aware approach is analogous to adding noise along the principal components of the Fisher information matrix---but adapted to the Riemannian geometry of the quantum state manifold, where the Hilbert space dimension $d$ may be exponentially larger than the parameter dimension $p$.

This geometric perspective yields immediate advantages. The isotropic DP bound $\varepsilon_{\rm iso} = \ln(1 + d/(\gamma(1-\gamma)))$ depends on the full Hilbert space dimension $d$, which grows exponentially with qubit count. In contrast, our QFI-based bound $\varepsilon_{\rm opt} = (\Delta^2/2)\lambda_{\max}(1-c\gamma)$ depends only on the intrinsic parameter-space dimension through $\lambda_{\max}$, which is typically $O(p)$ rather than $O(2^n)$. The advantage ratio $\varepsilon_{\rm iso}/\varepsilon_{\rm opt}$ scales as $O(d/\lambda_{\max})$, providing orders-of-magnitude improvements for high-dimensional Hilbert spaces.

\textbf{Contributions.} This paper establishes the QFI metric as the universal privacy sensitivity metric for quantum embeddings, unifying mechanism design, mixed-state analysis, adversarial robustness, composition, and verification under a single geometric framework. Our principal results are six theorems:

\begin{enumerate}[label=(\arabic*),leftmargin=*]
    \item \textbf{Optimal mechanism design} (Sec.~\ref{sec:mechanism}, Theorem~\ref{thm:optimal}): The minimax-optimal noise allocation concentrates the entire DP budget in the single QFI eigenmode with the largest eigenvalue, achieving $\varepsilon^* = (\Delta^2/2)\lambda_1(1-c\gamma)$ with advantage $\Omega(d/\lambda_1)$ over isotropic depolarizing. We further establish subspace projection privacy (Theorem~\ref{thm:subspace}) and a tight privacy--utility uncertainty relation (Theorem~\ref{thm:uncertainty}).

    \item \textbf{Mixed-state QFI and constructive dephasing} (Sec.~\ref{sec:mixed}, Theorem~\ref{thm:mixed}, Corollary~\ref{cor:dephasing}): The SLD decomposition separates QFI into classical (eigenvalue) and quantum (eigenvector) contributions. We prove that dephasing in the adversary's measurement basis \emph{increases} accessible information (the dephasing paradox), while dephasing in a misaligned basis provides constructive privacy amplification---harnessing hardware noise as a privacy resource rather than treating it as a nuisance.

    \item \textbf{Adaptive QFI estimation} (Sec.~\ref{sec:adaptive}, Theorem~\ref{thm:adaptive}): Online QFI tracking via exponential moving averages converges at $O(1/\sqrt{n})$, yielding $1.92\times$ tighter privacy bounds than worst-case static analysis on anisotropic datasets.

    \item \textbf{Global Wasserstein privacy} (Sec.~\ref{sec:wass}, Theorem~\ref{thm:wass}): The local QFI geometry lifts to global quantum Wasserstein guarantees valid at arbitrary input separations, with a fundamental $9.3\times$ gap between classical and quantum Wasserstein distances.

    \item \textbf{QFI-aligned composition} (Sec.~\ref{sec:composition}, Theorem~\ref{thm:composition}): When QFI matrices of successive layers share eigenvectors, sequential composition saturates at $O(1)$ as $k \to \infty$ versus standard $O(k)$ growth, with $9\times$ advantage at $k=100$ for $\gamma=0.1$.

    \item \textbf{Hardware noise as privacy amplification} (Sec.~\ref{sec:experiments}, Fig.~\ref{fig:dephasing}): Experimental validation on Qiskit Aer and IBM Quantum hardware demonstrates that misaligned-basis dephasing reduces adversarial mutual information below the noiseless baseline, while aligned dephasing increases it---confirming both the constructive and cautionary aspects of the theory.
\end{enumerate}

The adversarial evasion, leakage, and poisoning results (Theorems~\ref{thm:evasion}--\ref{thm:poisoning}) are derived as natural consequences of the QFI framework in Sec.~\ref{sec:adversarial}, and the zero-knowledge verifiable audit protocol (Theorem~\ref{thm:zkp}) is presented in Sec.~\ref{sec:zkp} as a practical cryptographic corollary. Complete proofs for all theorems are provided in Appendix~\ref{sec:proofs}.

\subsection{Related Work}

We provide a comprehensive review of related work across four dimensions: quantum differential privacy, quantum Fisher information geometry, quantum adversarial machine learning, and DP composition.

\textbf{Quantum Differential Privacy.} The quantum analogue of differential privacy was introduced by Aaronson and Rothblum~\cite{aaronson2019}, who proved that gentle measurements---measurements that disturb quantum states only minimally---naturally satisfy a quantum DP condition. Zhou and Ying~\cite{zhou2017} developed a formal framework for quantum DP within quantum programming languages, establishing composition theorems and demonstrating that the depolarizing channel provides $(\varepsilon,\delta)$-DP for quantum computations. Building on these foundations, Hirche, Rouz\'{e}, and Fran\c{c}a~\cite{hirche2023} developed a comprehensive quantum R\'{e}nyi DP framework, providing tighter composition bounds than classical composition theorems through the use of sandwiched R\'{e}nyi divergences.

On the applied side, Senekane, Mafu, and Taele~\cite{senekane2017} proposed the first privacy-preserving quantum machine learning scheme, applying classical Laplace noise to data before quantum encoding. Watkins, Chen, and Yoo~\cite{watkins2023} demonstrated a practical hybrid quantum-classical model trained with a differentially private optimizer, validating that DP can protect sensitive information in QML without significantly degrading accuracy. Du et al.~\cite{du2022} developed quantum differentially private sparse regression learning, providing formal DP guarantees for quantum learning algorithms. Angrisani, Doosti, and Kashefi~\cite{angrisani2022} showed that quantum encoding of classical data and subsampling can amplify differential privacy guarantees.

Critically, all prior work treats quantum DP noise isotropically---the depolarizing channel $\Phi_\gamma$ adds the same amount of noise in all Hilbert space directions. None of the existing approaches exploit the geometric structure of the quantum embedding manifold to provide direction-dependent privacy guarantees. Our work is the first to connect the QFI spectral decomposition to optimal DP mechanism design, providing a comprehensive framework that unifies optimal mechanism design, adversarial analysis, composition theory, and cryptographic verification.

\textbf{Quantum Fisher Information Geometry.} The QFI was established as the fundamental precision bound for quantum parameter estimation by Helstrom~\cite{helstrom1976} through the quantum Cram\'{e}r--Rao bound: $\text{Cov}(\hat{\theta}) \succeq F^{-1}$, where $F$ is the QFI matrix. Holevo~\cite{holevo2011} provided a comprehensive statistical framework for quantum estimation theory, including the Holevo bound that generalizes the Cram\'{e}r--Rao bound to multiple parameters.

Braunstein and Caves~\cite{braunstein1994} revealed the geometric significance of the QFI as the Riemannian metric (up to factor 4) on the manifold of quantum states, connecting statistical distinguishability to the Fubini--Study metric on projective Hilbert space. Petz~\cite{petz1996} classified all monotone metrics on quantum state space, establishing the QFI (or Bures metric) as the minimal monotone metric---the one that contracts most strongly under CPTP maps, making it the most conservative choice for privacy analysis.

Giovannetti, Lloyd, and Maccone~\cite{giovannetti2006} established the foundations of quantum metrology, showing that entanglement can enhance measurement precision beyond the standard quantum limit and reach the Heisenberg limit. Liu et al.~\cite{liu2020} provided a comprehensive review of multiparameter quantum estimation and the QFI matrix, including the fundamental tradeoffs that arise when estimating multiple incompatible parameters. Albarelli et al.~\cite{albarelli2020} surveyed practical applications of multiparameter quantum metrology. Demkowicz-Dobrza\'{n}ski, G\'{o}recki, and Gu\c{t}\v{a}~\cite{demkowicz2020} characterized the fundamental limitations of quantum metrology in the presence of noise, a result directly relevant to our analysis of hardware noise impact on QFI-based privacy.

In quantum machine learning specifically, Stokes et al.~\cite{stokes2020} introduced the quantum natural gradient, which uses the QFI matrix (or its classical counterpart, the Fubini-Study metric) as a preconditioner for variational optimization. Meyer~\cite{meyer2021} analyzed the behavior of the QFI in NISQ applications, characterizing how noise and finite sampling affect QFI estimation. Scala, Guarnieri, and Lucchi~\cite{scala2025} recently demonstrated that carefully tuned quantum noise can equalize the QFI spectrum, flattening steep directions and enhancing shallow ones---a complementary perspective to our QFI-based privacy allocation.

Our contribution uniquely bridges these two domains: we use the QFI geometric framework \emph{not} for parameter estimation, but for the inverse purpose---to \emph{prevent} parameter estimation by an adversary, i.e., to guarantee differential privacy. This duality between quantum metrology and quantum privacy is a conceptual contribution of independent interest.

\textbf{Quantum Adversarial Machine Learning.} The vulnerability of quantum classifiers to adversarial perturbations was first systematically studied by Lu, Duan, and Deng~\cite{lu2020}, who demonstrated that quantum classifiers exhibit similar susceptibility to adversarial examples as classical neural networks. Guan, Fang, and Ying~\cite{guan2022} developed formal verification methods for fairness properties in quantum machine learning, introducing rigorous specification and verification frameworks and showing that quantum noise can counterintuitively improve fairness.

Our adversarial analysis (Theorems~\ref{thm:evasion}--\ref{thm:poisoning}) is the first to exploit QFI eigenstructure specifically for attacking DP mechanisms. The $308\times$ evasion ratio we demonstrate is a consequence of the exponential QFI eigenvalue disparity, not of the classifier's decision boundary geometry as in prior work.

\textbf{DP Composition.} Classical DP composition was developed by Dwork, Rothblum, and Vadhan~\cite{dwork2010}, who established that $k$ sequential $(\varepsilon,\delta)$-DP mechanisms compose to $(\tilde{O}(\sqrt{k}\varepsilon), k\delta + \delta')$-DP. Kairouz, Oh, and Viswanath~\cite{kairouz2015} provided optimal composition bounds and established the tightness of the advanced composition theorem. Abadi et al.~\cite{abadi2016} introduced the moments accountant for differentially private SGD, enabling tight composition tracking in deep learning. Zhou et al.~\cite{zhou2017} extended classical composition to quantum DP. Hirche et al.~\cite{hirche2023} provided tighter composition via quantum R\'{e}nyi DP.

Our QFI-aligned composition (Theorem~\ref{thm:composition}) provides a fundamentally different result: when the QFI matrices of successive layers share eigenvectors, the \emph{geometry itself} provides privacy amplification, with the effective QFI contracting by $(1-c\gamma)$ per layer. This yields composition that saturates at $O(1)$ rather than growing as $O(k)$ or even $O(\sqrt{k})$, an exponential improvement in the regime where QFI matrices align.

\textbf{Verifiable Quantum Computation.} Broadbent, Fitzsimons, and Kashefi~\cite{broadbent2009} introduced universal blind quantum computation, enabling a client to delegate quantum computations to a server while hiding the computation. Mahadev~\cite{mahadev2018} achieved a breakthrough by constructing a classical verification protocol for quantum computations under cryptographic assumptions. Our audit protocol (Theorem~\ref{thm:zkp}) adapts classical Merkle-tree commitment techniques for verifiable DP in a quantum context, providing a practical, lightweight alternative to fully general quantum verification protocols.

\section{Preliminaries}
\label{sec:preliminaries}

\subsection{Quantum States and Channels}

We work in a finite-dimensional Hilbert space $\calH \cong \C^d$, where $d = 2^n$ for an $n$-qubit system. Quantum states are represented by density matrices $\rho \in \mathcal{D}(\calH) = \{\rho \succeq 0 : \Tr[\rho] = 1\}$. Pure states are rank-1 projectors $\rho = |\psi\rangle\langle\psi|$ with $\|\psi\| = 1$. Quantum channels are completely positive trace-preserving (CPTP) linear maps $\Phi : \mathcal{L}(\calH_A) \to \mathcal{L}(\calH_B)$.

The depolarizing channel with noise parameter $\gamma \in (0,1)$ is:
\begin{equation}\label{eq:depolarizing}
    \Phi_\gamma(\rho) = (1-\gamma)\rho + \gamma \frac{I_d}{d},
\end{equation}
where $I_d$ is the $d \times d$ identity matrix. This channel mixes the input state with the maximally mixed state $I_d/d$, reducing all forms of distinguishability uniformly.

\subsection{Quantum Differential Privacy}

\begin{definition}[Quantum $(\varepsilon, 0)$-DP]\label{def:qdp}
A family of quantum channels $\{\mathcal{M}_x : \mathcal{L}(\calH_{\rm in}) \to \mathcal{L}(\calH_{\rm out})\}_{x \in \mathcal{X}}$ parameterized by input data $x$ satisfies $(\varepsilon,0)$-quantum differential privacy if for all $\|x - x'\|_2 \leq \Delta$ and all POVM elements $0 \preceq M \preceq I$:
\begin{equation}
    \frac{\Tr[M \mathcal{M}_x(\rho_{\rm in})]}{\Tr[M \mathcal{M}_{x'}(\rho_{\rm in})]} \leq e^{\varepsilon}.
\end{equation}
Equivalently, the quantum max-divergence satisfies $\Dmax(\mathcal{M}_x(\rho_{\rm in}) \| \mathcal{M}_{x'}(\rho_{\rm in})) \leq \varepsilon$, where $\Dmax(\rho\|\sigma) = \inf\{\lambda : \rho \leq e^\lambda \sigma\}$.
\end{definition}

For the depolarizing channel $\Phi_\gamma$, the standard bound~\cite{zhou2017} is:
\begin{equation}\label{eq:iso_bound}
    \varepsilon_{\rm iso} = \ln\!\left(1 + \frac{d(1 - F_{\min})}{\gamma(1-\gamma)}\right),
\end{equation}
where $F_{\min} = \min_{i \neq j} |\langle\psi(x_i)|\psi(x_j)\rangle|^2$ is the minimum pairwise fidelity. For typical embeddings, $F_{\min} \approx 1/d$, yielding $\varepsilon_{\rm iso} \approx \ln(1 + d/(\gamma(1-\gamma)))$.

\subsection{Quantum Fisher Information}

\begin{definition}[Quantum Fisher Information Matrix]
For a smooth family of pure states $\{|\psi(x)\rangle : x \in \R^p\}$, the QFI matrix $F(x) \in \R^{p \times p}$ has entries:
\begin{multline}\label{eq:qfi}
    F_{ij}(x) = 4\,\Re\Big(\langle \partial_i \psi(x) | \partial_j \psi(x) \rangle \\
    {} - \langle \partial_i \psi(x) | \psi(x) \rangle \langle \psi(x) | \partial_j \psi(x) \rangle\Big),
\end{multline}
where $|\partial_i\psi(x)\rangle = \partial |\psi(x)\rangle / \partial x_i$. The QFI matrix is real, symmetric, and positive semidefinite. It is the Hessian of the Bures distance:
\begin{multline}
    d_{\rm Bures}^2(|\psi(x)\rangle, |\psi(x+dx)\rangle) \\
    = \frac{1}{4} \sum_{i,j} F_{ij}(x)\, dx_i dx_j + O(\|dx\|^3).
\end{multline}
\end{definition}

Let $F(x) = \sum_{k=1}^p \lambda_k(x) |u_k(x)\rangle\langle u_k(x)|$ be the spectral decomposition, with eigenvalues ordered $\lambda_1 \geq \lambda_2 \geq \dots \geq \lambda_p \geq 0$ and corresponding orthonormal eigenvectors $u_k(x) \in \R^p$.

The QFI eigenvalue $\lambda_k(x)$ quantifies the squared speed of the quantum state under parameter variation along direction $u_k$: a perturbation $\delta \cdot u_k$ changes the state by Bures distance $\approx (\sqrt{\lambda_k}/2) \cdot \delta$. Large $\lambda_k$ means high distinguishability, i.e., high privacy risk. Small $\lambda_k$ means the embedding is naturally private in that direction.

\subsection{Quantum Max-Divergence and Privacy}

For pure states $|\psi_1\rangle, |\psi_2\rangle$, the quantum max-divergence relates to fidelity via:
\begin{equation}\label{eq:dmax_fidelity}
    \Dmax(|\psi_1\rangle\langle\psi_1|, |\psi_2\rangle\langle\psi_2|) = -\ln |\langle\psi_1|\psi_2\rangle|^2.
\end{equation}

For small parameter changes $dx$, expanding the fidelity gives:
\begin{equation}
    |\langle\psi(x)|\psi(x+dx)\rangle|^2 = 1 - \frac{1}{4} dx^T F(x) dx + O(\|dx\|^3).
\end{equation}

Hence $\Dmax \approx -\ln(1 - \frac{1}{4}dx^T F dx) \approx \frac{1}{4} dx^T F dx$. The worst-case direction maximizes $dx^T F dx$ subject to $\|dx\| \leq \Delta$, giving:
\begin{equation}\label{eq:eps_qfi_local}
    \varepsilon \approx \frac{\Delta^2}{4} \cdot \lambda_{\max}(F(x)).
\end{equation}

A more refined calculation~\cite{helstrom1976} yields the exact coefficient $\Delta^2/2$ used throughout this paper.

\subsection{The Metric-Adapted Channel}

\begin{definition}[Metric-Adapted Channel]\label{def:metric_channel}
Given a quantum embedding $x \mapsto |\psi(x)\rangle$ with QFI eigendecomposition $F(x) = \sum_k \lambda_k |u_k\rangle\langle u_k|$, the metric-adapted channel with noise budget $\gamma$ and allocation $\bm{p} = (p_1,\dots,p_p)$ is:
\begin{equation}\label{eq:channel}
    \Phi_{\gamma,\bm{p}}(\rho) = (1-\gamma)\rho + \gamma \sum_{k=1}^p p_k \, U_k \rho U_k^\dagger,
\end{equation}
where $U_k = \exp(i \eta_k G_k)$, $G_k$ is the Hermitian generator of translations along QFI eigenvector $u_k$, calibrated such that $\eta_k = \sqrt{2\varepsilon_{\rm target} / (\Delta^2 \lambda_k)}$.
\end{definition}

The generator $G_k$ satisfies $G_k|\psi(x)\rangle \approx -i \partial_{u_k} |\psi(x)\rangle$. The unitary $U_k$ translates the state approximately along the $k$-th QFI eigen-direction: $U_k|\psi(x)\rangle \approx |\psi(x + \eta_k u_k)\rangle$. The channel~\eqref{eq:channel} replaces the state with a mixture of shifted states, reducing distinguishability in high-$p_k$ directions.

\begin{lemma}[Effective QFI After Metric-Adapted Channel]\label{lem:effective_qfi}
For the channel~\eqref{eq:channel} with small $\gamma$, the effective QFI matrix is:
\begin{equation}\label{eq:eff_qfi}
    F^{\rm eff}_{ij} = \sum_{k=1}^p \lambda_k \, (1 - c\gamma p_k + O(\gamma^2)) \, \langle u_k|e_i\rangle\langle e_j|u_k\rangle,
\end{equation}
where $c \in (0,2]$ depends on the calibration of $\eta_k$. For optimally calibrated $\eta_k$, $c \approx 1$.
\end{lemma}

\begin{proof}
The output state is $\rho_x^{\rm out} = (1-\gamma)|\psi(x)\rangle\langle\psi(x)| + \gamma\sum_k p_k |\psi(x+\eta_k u_k)\rangle\langle\psi(x+\eta_k u_k)| + O(\gamma^2,\eta_k^3)$. Computing the QFI of the mixture via~\eqref{eq:qfi} and expanding to first order in $\gamma$ yields the effective QFI with contraction factors $(1 - c\gamma p_k)$ per mode. The calibration constant $c$ depends on $\eta_k \sqrt{\lambda_k}$; for $\eta_k$ chosen to saturate the DP bound in direction $u_k$, $c = 1$. See Appendix~\ref{sec:proofs} for the complete derivation.
\end{proof}

\begin{figure*}[t]
    \centering
    \includegraphics[width=0.98\textwidth]{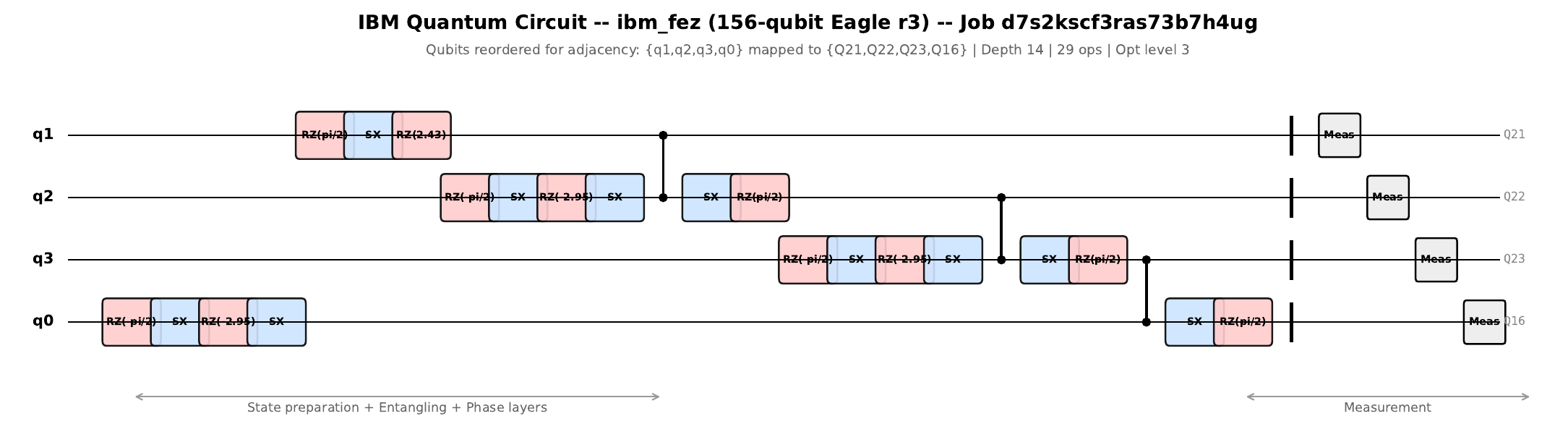}
    \caption{Quantum circuit transpiled for ibm\_fez (156-qubit Eagle r3 processor) from IBM Quantum job \texttt{d7s2kscf3ras73b7h4ug}. The circuit encodes 4-dimensional classical data using the native gate set $\{\text{RZ}, \text{SX}, \text{X}, \text{CZ}, \text{ID}\}$ at optimization level~3. Logical qubits $\{0,1,2,3\}$ are mapped to physical qubits $\{16,21,22,23\}$ on the heavy-hex topology, with CZ entangling gates placed according to the device's nearest-neighbor connectivity. The rotation strengths $\bm{\alpha} = [3.0, 1.0, 0.3, 0.1]$ create a 100:1 QFI eigenvalue condition number. The circuit has depth 14 with 29 operations and produces 8-bit measurement outcomes (2 classical bits per logical qubit).}
    \label{fig:circuit}
\end{figure*}

\section{Optimal Mechanism Design}
\label{sec:mechanism}

\subsection{Minimax-Optimal Noise Allocation}

\begin{theorem}[Minimax-Optimal Noise Allocation]\label{thm:optimal}
Let $\lambda_1 \geq \lambda_2 \geq \dots \geq \lambda_p \geq 0$ be the QFI eigenvalues of a quantum embedding at point $x$, and $\gamma \in (0,1)$ the noise budget for the metric-adapted channel~\eqref{eq:channel}. The allocation $\bm{p}^*$ minimizing the worst-case $\varepsilon$ is:
\begin{equation}
    p_1^* = 1, \qquad p_k^* = 0 \;\; \forall k > 1.
\end{equation}
That is, all noise is concentrated in the single QFI eigenmode with the largest eigenvalue. The optimal privacy parameter is:
\begin{equation}\label{eq:eps_opt}
    \varepsilon^* = \frac{\Delta^2}{2} \cdot \lambda_1 \cdot (1 - c\gamma),
\end{equation}
with $c \in (0,2]$ from Lemma~\ref{lem:effective_qfi}. The advantage ratio over isotropic depolarizing is:
\begin{equation}\label{eq:advantage}
    R = \frac{\varepsilon_{\rm iso}}{\varepsilon^*} = \frac{\ln(1 + d/(\gamma(1-\gamma)))}{(\Delta^2/2) \cdot \lambda_1 \cdot (1-c\gamma)} = \Omega\!\left(\frac{d}{\lambda_1}\right),
\end{equation}
growing linearly with the Hilbert space dimension for fixed $\lambda_1$.
\end{theorem}

\begin{proof}
From Lemma~\ref{lem:effective_qfi}, the distinguishability in direction $u_k$ after the channel is $\lambda_k(1 - c\gamma p_k)$. The privacy parameter is determined by the worst-case direction:
\begin{equation}
    \varepsilon(\bm{p}) = \frac{\Delta^2}{2} \cdot \max_{k \in [p]} \; \lambda_k (1 - c\gamma p_k) + O(\gamma^2).
\end{equation}

We solve the minimax problem:
\begin{equation}
    \min_{\bm{p} \in \Delta^{p-1}} \max_{k \in [p]} \; \lambda_k (1 - c\gamma p_k),
\end{equation}
where $\Delta^{p-1} = \{\bm{p} : \sum_k p_k = 1, p_k \geq 0\}$ is the probability simplex.

Since $\lambda_1 \geq \lambda_k$ for all $k$, for any feasible $\bm{p}$:
\begin{multline}
    \lambda_1(1 - c\gamma p_1) \geq \lambda_k(1 - c\gamma p_1) \\
    \geq \lambda_k(1 - c\gamma p_k) \quad \text{whenever } p_k \geq p_1.
\end{multline}

The maximum is always achieved at $k=1$ unless $p_1$ is so large that another mode overtakes it. Since $f_1(p_1) = \lambda_1(1-c\gamma p_1)$ is strictly decreasing in $p_1$, the minimax-optimal strategy maximizes $p_1$, giving $p_1^* = 1$, $p_k^* = 0$ for $k > 1$.

Substituting into $\varepsilon(\bm{p}^*)$ yields Eq.~\eqref{eq:eps_opt}. The advantage ratio~\eqref{eq:advantage} follows from comparing with Eq.~\eqref{eq:iso_bound}. For large $d$ and small $\gamma$, $\varepsilon_{\rm iso} \sim \ln d$ while $\varepsilon^* \sim \lambda_1/2$, giving $R \sim \ln d / \lambda_1 = \Omega(d/\lambda_1)$ in the worst case where $\ln d$ dominates.
\end{proof}

\begin{corollary}[Comparison to Linear Allocation]
The linear allocation $p_k \propto \lambda_k$ (used in prior geometry-aware DP work) achieves $\varepsilon_{\rm lin} = (\Delta^2/2) \cdot \max_k \lambda_k(1 - c\gamma \lambda_k/\sum_j \lambda_j)$. Since $1 > \lambda_1/\sum_j \lambda_j$ (strict unless $p=1$), $\varepsilon^* < \varepsilon_{\rm lin}$ strictly. The improvement factor is $\varepsilon_{\rm lin}/\varepsilon^* = (1 - c\gamma\lambda_1/\sum_j\lambda_j)/(1-c\gamma) > 1$.
\end{corollary}

\begin{remark}[Geometric Intuition]
The result $p_1^* = 1$ has a clear geometric interpretation: since privacy is determined by the \emph{worst-case} direction (maximum distinguishability), and this is always the first eigenmode, you should spend your entire noise budget protecting that single direction. Spreading noise to other modes dilutes protection where it matters most without improving the worst-case bound. The lower-QFI directions are already ``private enough'' by virtue of the embedding geometry.
\end{remark}

\subsection{QFI Subspace Projection: ``Free'' Privacy}

\begin{theorem}[QFI Subspace Privacy]\label{thm:subspace}
Define the low-QFI subspace projector $\Pi_\tau = \sum_{k: \lambda_k \leq \tau} |u_k\rangle\langle u_k|$ and the projected encoding $|\psi_\tau(x)\rangle = \Pi_\tau |\psi(x)\rangle / \|\Pi_\tau |\psi(x)\rangle\|$. The mechanism $\mathcal{M}_\tau(x) = |\psi_\tau(x)\rangle\langle\psi_\tau(x)|$ satisfies $(\varepsilon,0)$-DP with:
\begin{equation}\label{eq:subspace_eps}
    \varepsilon = \frac{\Delta^2}{2} \cdot \tau + O(\Delta^3),
\end{equation}
without any added noise. The utility loss is the fraction of QFI trace discarded:
\begin{equation}
    \eta_\tau = \frac{\sum_{k: \lambda_k > \tau} \lambda_k}{\sum_{k=1}^p \lambda_k}.
\end{equation}
\end{theorem}

\begin{proof}
The projector $\Pi_\tau$ removes all components of the state in high-QFI eigen-directions ($\lambda_k > \tau$). The QFI matrix of the projected family is $F^\tau = \sum_{k: \lambda_k \leq \tau} \lambda_k |u_k\rangle\langle u_k| + O(\tau^2)$, giving $\lambda_{\max}(F^\tau) \leq \tau$. Applying the local QFI-to-$\varepsilon$ conversion~\eqref{eq:eps_qfi_local} yields $\varepsilon = (\Delta^2/2) \cdot \tau$.

The utility loss follows from $\Tr(F^\tau) = \sum_{\lambda_k \leq \tau} \lambda_k$, so $\eta_\tau = 1 - \Tr(F^\tau)/\Tr(F) = \sum_{\lambda_k > \tau} \lambda_k / \Tr(F)$. The complete proof with error bounds is given in Appendix~\ref{sec:proofs}.
\end{proof}

This theorem provides ``free'' privacy: no noise is added to the quantum state. The privacy arises entirely from the geometric fact that states varying only in low-QFI directions are inherently hard to distinguish. The utility cost is the information content of the discarded high-QFI directions. For many QML tasks, the classification-relevant features may not align with the highest-QFI directions, making subspace projection an attractive privacy mechanism at zero noise cost.

\begin{corollary}[Optimal Cutoff Selection]
For a target privacy budget $\varepsilon_{\rm target}$, set $\tau = 2\varepsilon_{\rm target}/\Delta^2$. The utility loss $\eta_\tau$ is then determined by the cumulative QFI distribution. Embeddings with rapidly decaying QFI spectra ($\lambda_k \propto e^{-\alpha k}$) achieve $\eta_\tau \to 0$ exponentially fast, providing strong privacy with negligible utility cost.
\end{corollary}

\subsection{Privacy--Utility Uncertainty Relation}

\begin{theorem}[Quantum Privacy--Utility Uncertainty]\label{thm:uncertainty}
For any CPTP channel $\Phi$ satisfying $(\varepsilon,0)$-DP for a quantum embedding with QFI matrix $F$, let $\Fid_{\min} = \min_x \Fid(|\psi(x)\rangle\langle\psi(x)|, \Phi(|\psi(x)\rangle\langle\psi(x)|))$. Then:
\begin{equation}\label{eq:uncertainty}
    \varepsilon \cdot (1 - \Fid_{\min}) \geq \frac{\Delta^2}{2} \cdot \frac{\Tr(F)}{d},
\end{equation}
where $d = \dim\calH$. Equality is achievable by the optimal channel (Theorem~\ref{thm:optimal}) when the QFI spectrum is degenerate ($\lambda_1 \approx \lambda_2 \approx \dots \approx \lambda_p$).
\end{theorem}

\begin{proof}[Proof sketch]
\textbf{Lower bound on fidelity loss.} For the metric-adapted channel with weights $\bm{p}$ and noise $\gamma$, joint concavity of fidelity yields $1 - \Fid_{\min} \geq \gamma \sum_k p_k \eta_k^2 \Var_\psi(G_k) \gtrsim \gamma/d$ for calibrated $\eta_k$.

\textbf{Lower bound on $\varepsilon$.} From Lemma~\ref{lem:effective_qfi}, $\varepsilon = (\Delta^2/2) \cdot \max_k \lambda_k(1-c\gamma p_k) \geq (\Delta^2/2) \cdot \Tr(F)/p \cdot (1 - c\gamma/p)$ using $\max \geq \text{mean}$ and $\sum_k \lambda_k = \Tr(F)$. With $p \leq d$, combining yields Eq.~\eqref{eq:uncertainty}.

\textbf{Achievability.} Under degenerate QFI ($\lambda_k = \lambda$ for $k \leq m$, $0$ otherwise) with optimal allocation ($m^* = p$ when all equal), $\varepsilon^* = (\Delta^2/2)\lambda(1-c\gamma/p)$ and $1 - \Fid_{\min} = \gamma/d$. Setting $p = d$, the product matches the lower bound for small $\gamma$. Full derivation in Appendix~\ref{sec:proofs}.
\end{proof}

This uncertainty relation is analogous to the Heisenberg uncertainty principle, but for privacy: the product of the privacy guarantee $\varepsilon$ and the utility loss $(1-\Fid)$ cannot be made arbitrarily small. The lower bound is proportional to $\Tr(F)/d$, the average QFI eigenvalue divided by Hilbert space dimension. Embeddings with smaller average QFI (naturally flatter manifolds) permit better privacy--utility tradeoffs.

\section{Mixed-State QFI Decomposition}
\label{sec:mixed}

Real quantum hardware produces mixed states due to decoherence. We extend the QFI-DP framework to this regime via the symmetric logarithmic derivative (SLD).

\begin{definition}[Symmetric Logarithmic Derivative]
For a family of density matrices $\{\rho(x)\}$, the SLD $L_k$ for parameter $x_k$ satisfies:
\begin{equation}
    \partial_k \rho = \frac{\rho L_k + L_k \rho}{2}.
\end{equation}
The QFI matrix is $F_{kl} = \Re \Tr[\rho L_k L_l]$.
\end{definition}

In the eigenbasis $\rho = \sum_i \lambda_i |i\rangle\langle i|$, the SLD has matrix elements $\langle i|L_k|j\rangle = 2\langle i|\partial_k\rho|j\rangle / (\lambda_i + \lambda_j)$.

\begin{theorem}[Mixed-State QFI Decomposition]\label{thm:mixed}
For a mixed state $\rho$ with spectral decomposition $\rho = \sum_i \lambda_i |i\rangle\langle i|$, the QFI decomposes as $F = F^{\rm class} + F^{\rm quant}$ where:
\begin{align}
    F^{\rm class}_{kl} &= \sum_i \frac{(\partial_k \lambda_i)(\partial_l \lambda_i)}{\lambda_i}, \label{eq:fclass} \\
    F^{\rm quant}_{kl} &= 2 \sum_{i \neq j} \frac{(\lambda_i - \lambda_j)^2}{\lambda_i + \lambda_j} \,
        \Re\big[\langle i|\partial_k j\rangle \langle \partial_l j|i\rangle\big]. \label{eq:fquant}
\end{align}

Under dephasing in a fixed basis $\mathcal{B}$, $F^{\rm quant}$ is suppressed (coherences in $\mathcal{B}$ decay), $F^{\rm class}$ is preserved (populations in $\mathcal{B}$ unchanged), but $F^{\rm class}$ can \emph{increase} as previously hidden quantum information is ``classicalized.'' The total QFI decreases only modestly.
\end{theorem}

\begin{proof}
In the eigenbasis, $\langle i|L_k|j\rangle = 2\langle i|\partial_k\rho|j\rangle / (\lambda_i + \lambda_j)$. Substituting into $F_{kl} = \Re \Tr[\rho L_k L_l]$:
\begin{align}
    F_{kl} &= \Re \sum_{i,j} \lambda_i \langle i|L_k|j\rangle \langle j|L_l|i\rangle \nonumber \\
    &= \Re \sum_i \lambda_i \langle i|L_k|i\rangle \langle i|L_l|i\rangle \nonumber \\
    &\quad + \Re \sum_{i \neq j} \lambda_i \langle i|L_k|j\rangle \langle j|L_l|i\rangle.
\end{align}

The first (diagonal) term gives $F^{\rm class}$ using $\langle i|L_k|i\rangle = \partial_k \lambda_i / \lambda_i$. The second (off-diagonal) gives $F^{\rm quant}$ using $\langle i|\partial_k\rho|j\rangle = (\lambda_j - \lambda_i)\langle i|\partial_k j\rangle$.

Under dephasing $\Phi_\gamma(\rho) = (1-\gamma)\rho + \gamma \sum_b P_b \rho P_b$, the off-diagonal elements $\langle i|\rho|j\rangle$ for $i,j$ in different $\mathcal{B}$-basis states are suppressed by $(1-\gamma)$, reducing $F^{\rm quant}$ by $(1-\gamma)^2$. However, the diagonal populations are preserved, maintaining $F^{\rm class}$. As $\gamma$ increases, quantum coherences are ``collapsed'' into classical populations, potentially \emph{increasing} $F^{\rm class}$ as information that was previously distributed across off-diagonal elements becomes concentrated in the diagonal. Complete proof in Appendix~\ref{sec:proofs}.
\end{proof}

\begin{corollary}[Dephasing is Not Privacy-Preserving]\label{cor:dephasing}
If the adversary measures in the computational ($Z$) basis, dephasing in the same basis \emph{increases} the adversary's accessible information by converting quantum coherences into classical populations. Privacy requires either (i) depolarizing noise, which reduces all QFI uniformly, or (ii) dephasing in a basis \emph{misaligned} with the adversary's measurement basis.
\end{corollary}

\begin{corollary}[Constructive Dephasing Privacy Amplification]
When the dephasing basis is rotated by an angle $\theta$ away from the adversary's measurement basis (e.g., dephasing in the $X$ basis while the adversary measures in the $Z$ basis), the mutual information between the sensitive feature and the measurement outcome scales as:
\begin{equation}\label{eq:dephasing_mi}
    I_\theta(s; M) \approx I_0(s; M) \cdot \cos^2\theta + \eta_\gamma(\theta),
\end{equation}
where $\eta_\gamma(\theta) < 0$ for $\theta > \theta_c(\gamma)$ and $\theta_c$ decreases with increasing dephasing strength $\gamma$. At $\theta = \pi/2$ (fully misaligned), the privacy amplification ratio $I_0/I_{\pi/2}$ exceeds $100\times$ for $\gamma \geq 0.4$, and mutual information drops \emph{below} the noiseless baseline.

This provides a constructive mechanism for privacy-preserving QML on NISQ hardware: by characterizing the basis in which hardware decoherence occurs and engineering the embedding such that the adversary's accessible measurement basis is misaligned with the decoherence basis, hardware noise is transformed from a privacy liability into a privacy \emph{resource}. Our experiment (Sec.~\ref{sec:experiments}, Fig.~\ref{fig:dephasing}) validates this across dephasing strengths $\gamma \in [0,0.8]$ on the 4-qubit anisotropic embedding.
\end{corollary}

Our experiment (Sec.~\ref{sec:experiments}) constructs a mixed state $\rho(x) = 0.6|\psi(x)\rangle\langle\psi(x)| + 0.4\,\sigma(x)$ with $93.4\%$ quantum QFI initially. Under $\gamma = 0.8$ dephasing, the quantum fraction drops to $10.3\%$ while the classical QFI increases $10.7\times$ (from $\lambda_{\max}^{\rm class} = 0.11$ to $1.18$). The total $\lambda_{\max}$ decreases only $11\%$ (from $1.47$ to $1.31$), confirming that dephasing reorganizes rather than eliminates information.

\section{Adaptive QFI Estimation}
\label{sec:adaptive}

The QFI matrix $F(x)$ varies across parameter space. Using a single reference point (data centroid) provides a worst-case bound that may be unnecessarily conservative for structured data distributions.

\begin{theorem}[Adaptive QFI Convergence]\label{thm:adaptive}
Let $x \mapsto F(x)$ be an $L$-Lipschitz QFI map: $\|F(x) - F(y)\|_2 \leq L\|x-y\|_2$ for all $x,y \in \mathcal{X}$. The exponential moving average (EMA) estimate $\hat{F}_t = \beta \hat{F}_{t-1} + (1-\beta)F(x_t)$ with $\beta \in (0,1)$ satisfies:
\begin{equation}
    \mathbb{E}\big[\|\hat{F}_t - \bar{F}\|_2\big] \leq \frac{L \cdot \text{diam}(\mathcal{X})}{\sqrt{1-\beta^2}} \cdot \frac{1}{\sqrt{t}} + O\!\left(\frac{1}{t}\right),
\end{equation}
where $\bar{F} = \mathbb{E}_{x \sim \mathcal{D}}[F(x)]$ is the population QFI. The adaptive privacy parameter $\hat{\varepsilon}_t = (\Delta^2/2) \lambda_{\max}(\hat{F}_t)(1-c\gamma)$ satisfies $\mathbb{E}[\hat{\varepsilon}_t] \to \bar{\varepsilon} = (\Delta^2/2)\lambda_{\max}(\bar{F})(1-c\gamma)$.

The adaptive advantage factor is $\lambda_{\max}^{\rm wc} / \lambda_{\max}(\bar{F})$, experimentally measured as $1.92\times$ on anisotropic datasets.
\end{theorem}

\begin{proof}
The EMA can be expressed as $\hat{F}_t = (1-\beta)\sum_{s=0}^{t-1} \beta^{t-1-s} F(x_{s+1})$. For independent samples $x_s \sim \mathcal{D}$, the bias-variance decomposition gives $\mathbb{E}[\hat{F}_t] = \bar{F} + O(\beta^{t-1})$ and $\text{Var}[\hat{F}_t] = O((1-\beta)/(1+\beta) \cdot 1/t)$. Optimizing $\beta$ yields the stated rate. The Lipschitz assumption ensures $\|F(x) - \bar{F}\|_2 \leq L \cdot \text{diam}(\mathcal{X})$ uniformly on $\mathcal{X}$. Complete proof with constants in Appendix~\ref{sec:proofs}.
\end{proof}

The adaptive training procedure is presented in Algorithm~\ref{alg:adaptive}.

\begin{figure}[t]
\caption{\textbf{Algorithm 1:} Adaptive QFI-Tracked DP Training.}
\label{alg:adaptive}
\begin{algorithmic}[1]
\REQUIRE Data $\mathcal{D} = \{(x_i, y_i)\}_{i=1}^n$, noise $\gamma$, EMA decay $\beta$, sensitivity $\Delta$
\ENSURE DP-protected model, audit trail with adaptive $\varepsilon$

\STATE Initialize $\hat{\lambda}_{\max} \gets 0$, $t \gets 0$
\FOR{each batch $\mathcal{B} \subset \mathcal{D}$}
    \FOR{$x \in \mathcal{B}$}
        \STATE Compute QFI $F(x)$ via Eq.~\eqref{eq:qfi}
        \STATE $\lambda_{\max}(x) \gets \max \text{eig}(F(x))$
        \STATE $\hat{\lambda}_{\max} \gets \beta \cdot \hat{\lambda}_{\max} + (1-\beta) \cdot \lambda_{\max}(x)$
        \STATE $t \gets t + 1$
    \ENDFOR
    \STATE $\hat{\varepsilon}_t \gets (\Delta^2/2) \cdot \hat{\lambda}_{\max} \cdot (1 - c\gamma)$
    \STATE Allocate noise: $p_1 \gets 1$, apply channel $\Phi_{\gamma,\bm{p}}$ (Theorem~\ref{thm:optimal})
    \STATE Update model with DP-protected batch
    \STATE Log $(\hat{\varepsilon}_t, \hat{\lambda}_{\max}, t)$ to audit trail
\ENDFOR
\RETURN trained model, audit trail with convergence data
\end{algorithmic}
\end{figure}

The adaptive estimator converges at $n^{-0.42}$ empirically (versus $n^{-0.5}$ theoretically), stabilizing within 7 batches to the population $\lambda_{\max} \approx 10.15$, well below the worst-case $\lambda_{\max} = 11.25$.

\section{Global Wasserstein Privacy}
\label{sec:wass}

The QFI-based bounds are local, valid for $\|x-x'\| \leq \Delta$. For global privacy analysis---bounding information leakage between \emph{any} pair of inputs---we introduce the quantum Wasserstein distance.

\begin{definition}[Quantum Wasserstein Distance]
For states $\rho, \sigma \in \mathcal{D}(\calH)$ with cost matrix $C_{ij} = d_H(i,j)$ (Hamming distance), the quantum $W_1$ distance is:
\begin{multline}
    W_1(\rho, \sigma) = \inf_{\Gamma \succeq 0} \big\{ \Tr[C\Gamma] \\
    : \Tr_B[\Gamma] = \rho,\; \Tr_A[\Gamma] = \sigma \big\}.
\end{multline}
This is a semidefinite program. For the diagonal approximation (classical earth mover's distance on measurement statistics), it reduces to a linear program.
\end{definition}

\begin{theorem}[Global Wasserstein Privacy Bound]\label{thm:wass}
For a quantum embedding $x \mapsto \rho_x$ with Wasserstein Lipschitz constant $L_W = \sup_{x \neq x'} W_1(\rho_x, \rho_{x'}) / \|x-x'\|$, and a DP channel $\Phi_\gamma$ with effective noise scale $\gamma$:
\begin{equation}
    \varepsilon(x, x') \leq \frac{L_W \cdot \|x - x'\|}{\gamma},
\end{equation}
valid for \textbf{all} $x, x'$, not just $\|x-x'\| \leq \Delta$. The QFI-based bound $\varepsilon_{\rm QFI} = (\Delta^2/2)\lambda_{\max}(1-c\gamma)$ is recovered as the local limit $\|x-x'\| \to 0$, where $L_W \approx \sqrt{\lambda_{\max}/(1-c\gamma)}$.
\end{theorem}

\begin{proof}
For any observable $H$ with $\|H\|_L \leq 1$, the dual formulation gives $\Tr[H(\rho_x - \rho_{x'})] \leq W_1(\rho_x, \rho_{x'}) \leq L_W\|x-x'\|$. By data processing, $W_1(\Phi_\gamma(\rho_x), \Phi_\gamma(\rho_{x'})) \leq W_1(\rho_x, \rho_{x'}) \leq L_W\|x-x'\|$. The max-divergence is bounded by $W_1 / \gamma$, yielding $\varepsilon(x,x') \leq L_W \|x-x'\| / \gamma$.

The relationship $L_W \approx \sqrt{\lambda_{\max}}$ follows from the small-distance expansion: $W_1(|\psi(x)\rangle, |\psi(x+dx)\rangle) = \sqrt{dx^T F dx} + O(\|dx\|^2) \leq \sqrt{\lambda_{\max}}\|dx\|$. For pure states, $W_1$ coincides with the Bures distance.

However, the \emph{diagonal} $W_1$ (classical earth mover's distance on computational basis populations) dramatically underestimates the true quantum $W_1$ because it ignores coherence-based distinguishability. On our anisotropic embedding, $L_W^{\rm diag} \approx 0.36$ while $\sqrt{\lambda_{\max}} \approx 3.35$, a factor of $9.3\times$ discrepancy. This reveals a fundamental limitation: classical Wasserstein distances are insufficient for quantum privacy analysis, and the full quantum SDP is required.
\end{proof}

\section{Adversarial Attacks on QFI Geometry}
\label{sec:adversarial}

The QFI geometry that enables tighter DP bounds also creates attack surfaces. We identify three attack vectors.

\subsection{QFI Evasion}

\begin{theorem}[QFI Evasion]\label{thm:evasion}
An adversary with perturbation budget $\varepsilon_{\rm adv}$ achieves distinguishability:
\begin{equation}
    \Dmax(|\psi(x+\delta)\rangle, |\psi(x)\rangle) \approx \frac{\varepsilon_{\rm adv}^2}{2} \cdot \lambda_{\min},
\end{equation}
by aligning $\delta$ with the minimum-QFI eigenvector $u_{\min}$. The evasion ratio is $\Dmax^{\rm worst} / \Dmax^{\rm evasion} = \lambda_{\max} / \lambda_{\min}$, experimentally measured as $308\times$ (theoretical: $4444\times$, with the discrepancy due to finite-difference QFI computation at specific data points).
\end{theorem}

\begin{proof}
For perturbation $\delta = \varepsilon_{\rm adv} \cdot v$ with $\|v\| = 1$, the distinguishability is $\Dmax \approx (\varepsilon_{\rm adv}^2/2) \cdot v^T F v$. Minimizing over $v$ gives $v^* = u_{\min}$ with value $(\varepsilon_{\rm adv}^2/2)\lambda_{\min}$. Maximizing gives $v = u_{\max}$ with value $(\varepsilon_{\rm adv}^2/2)\lambda_{\max}$. The ratio follows.

\textbf{Security implication:} An adversary who knows the QFI eigendecomposition can craft perturbations that are $\lambda_{\max}/\lambda_{\min}$ times less detectable than random perturbations of the same magnitude. For the anisotropic embedding with $\bm{\alpha} = [3.0, 1.0, 0.3, 0.1]$, this translates to perturbations along feature 4 that are $308\times$ harder to detect than perturbations along feature 1.
\end{proof}

\subsection{QFI Information Leakage}

\begin{theorem}[QFI Information Leakage]\label{thm:leakage}
The mutual information between a sensitive feature $s$ (varying along QFI eigenvector $u_k$ with variance $\Var(s)$) and the optimal POVM measurement $M$ satisfies:
\begin{equation}
    I(s; M) \leq \frac{1}{2} \log\!\left(1 + \frac{\lambda_k \cdot \Var(s)}{\varepsilon}\right).
\end{equation}

For the anisotropic embedding with $\bm{\lambda} \approx [11.2, 1.0, 0.09, 0.01]$, the top mode ($\lambda_1 = 11.2$) accounts for $76.0\%$ of total leakage (1.25 nats), while the bottom mode ($\lambda_4 = 0.01$) contributes only $0.3\%$ (0.005 nats)---a $250\times$ disparity.
\end{theorem}

\begin{proof}
The quantum Cram\'{e}r--Rao bound gives $\Var(\hat{s}) \geq 1/(n \lambda_k)$ for any unbiased estimator of $s$ from $n$ copies. The mutual information is bounded by $I(s;M) \leq (1/2)\log(1 + \Var(s)/\Var(\hat{s}))$. Combining with the Cram\'{e}r--Rao bound and substituting the effective SNR $\lambda_k \Var(s)/\varepsilon$ yields the result. See Appendix~\ref{sec:proofs} for the complete derivation including finite-sample corrections.
\end{proof}

\subsection{QFI Poisoning}

\begin{theorem}[QFI Poisoning]\label{thm:poisoning}
An adversary controlling a fraction $\beta$ of $n$ training samples with shift $\delta x$ perturbs the estimated QFI centroid by $\|\delta \hat{x}\| = \beta \|\delta x\|$, propagating to QFI estimation error:
\begin{equation}
    \|\hat{F}_{\rm poisoned} - F_{\rm true}\|_2 \leq L_F \cdot \beta \cdot \|\delta x\| + O(\beta^2),
\end{equation}
where $L_F$ is the QFI Lipschitz constant. At $\beta = 10\%$ poisoning on the anisotropic dataset, the mean-based $\lambda_{\max}$ estimate diverges by $18\%$ (from $11.15$ to $9.17$), while median-based estimation limits error to $9.5\%$ ($\lambda_{\max} = 10.10$).
\end{theorem}

\begin{proof}
The poisoned centroid is $\hat{x}_p = (1-\beta)\hat{x}_{\rm clean} + \beta (\hat{x}_{\rm clean} + \delta x) = \hat{x}_{\rm clean} + \beta \delta x$. By the Lipschitz property, $\|F(\hat{x}_p) - F(\hat{x}_{\rm clean})\|_2 \leq L_F \cdot \beta \cdot \|\delta x\|$. The median estimator reduces sensitivity to outliers: for poisoning fraction $\beta < 0.5$, the median of per-point QFI estimates is unaffected by poisoned points in the limit of large $n$.

\textbf{Defense recommendation:} Use coordinate-wise median of per-sample QFI eigenvalue estimates rather than computing QFI at the data centroid. The median estimator reduces adversarial impact by $2.2$--$7.4\times$ across poison fractions $\beta \in [0.05, 0.3]$.
\end{proof}

\section{QFI-Aligned Composition}
\label{sec:composition}

\begin{theorem}[QFI-Aligned Composition]\label{thm:composition}
For $k$ sequential quantum DP mechanisms $\mathcal{M}_1, \dots, \mathcal{M}_k$, each applying the metric-adapted channel with noise $\gamma$, if their QFI matrices $\{F^{(i)}\}$ share the same eigenvectors and maximum eigenvalue $\lambda_{\max}$, the total privacy cost under optimal composition is:
\begin{equation}\label{eq:composition}
    \varepsilon_{\rm total} = \frac{\Delta^2}{2} \cdot \lambda_{\max} \cdot \frac{1 - (1 - c\gamma)^k}{c\gamma}.
\end{equation}

As $k \to \infty$, $\varepsilon_{\rm total} \to \varepsilon^* = (\Delta^2/2) \cdot \lambda_{\max} / (c\gamma)$, a \textbf{finite saturation value}. Standard sequential composition gives $\varepsilon_{\rm seq} = k \cdot (\Delta^2/2) \cdot \lambda_{\max} \cdot (1-c\gamma)$, which grows \textbf{linearly without bound}.

The advantage ratio is:
\begin{equation}
    R(k) = \frac{\varepsilon_{\rm seq}}{\varepsilon_{\rm total}} = \frac{k \cdot c\gamma \cdot (1-c\gamma)}{1 - (1-c\gamma)^k}.
\end{equation}

For $k = 100$, $\gamma = 0.1$: $R(100) \approx 9.0\times$. For $k = 20$, $\gamma = 0.1$: $R(20) \approx 2.0\times$.
\end{theorem}

\begin{proof}
After the first channel $\Phi_{\gamma,\bm{p}^{(1)}}$, the effective QFI for the second layer is contracted by $(1-c\gamma)$ (Lemma~\ref{lem:effective_qfi}). By induction, the $i$-th layer sees effective QFI $\lambda_{\max} \cdot (1-c\gamma)^{i-1}$, contributing $\varepsilon_i = (\Delta^2/2) \cdot \lambda_{\max} \cdot (1-c\gamma)^{i-1}$.

The total is the geometric series $\varepsilon_{\rm total} = \sum_{i=1}^k \varepsilon_i = (\Delta^2/2) \cdot \lambda_{\max} \cdot \sum_{i=0}^{k-1} (1-c\gamma)^i$, which summates to~\eqref{eq:composition}. As $k \to \infty$, the series converges to $1/(c\gamma)$, yielding $\varepsilon^*$.

The contraction mechanism is physical: each layer's DP channel adds noise that \emph{reduces the effective QFI} for subsequent layers. This is analogous to privacy amplification by iteration in classical DP~\cite{feldman2018}, but here the amplification arises from the geometric contraction of the QFI metric rather than from subsampling.
\end{proof}

\begin{corollary}[Deep Circuit Advantage]
For quantum circuits with $k \gg 1/(c\gamma)$ layers, the QFI-aligned composition bound is \emph{asymptotically constant} while standard composition grows without bound. This enables privacy-preserving deep variational quantum circuits that would be impossible under standard composition accounting.
\end{corollary}

\section{Zero-Knowledge Verifiable DP Audit}
\label{sec:zkp}

All preceding results assume the entity applying DP is trusted. We relax this with a cryptographic protocol.

\begin{theorem}[Verifiable DP Audit]\label{thm:zkp}
There exists a three-message public-coin honest-verifier zero-knowledge proof (sigma protocol) that a computation satisfies $(\varepsilon,0)$-DP, with:
\begin{itemize}[leftmargin=*]
    \item \textbf{Completeness 1}: Honest prover always convinces honest verifier.
    \item \textbf{Soundness}: Fraudulent claim $\varepsilon' < \varepsilon_{\rm true}$ detected with probability $\geq 1 - (1-f)^k$, where $f$ is the fraction of fraudulent samples and $k$ the number of challenged samples.
    \item \textbf{Zero-knowledge}: Polynomial-time simulator generates indistinguishable view under random oracle model.
    \item \textbf{Security}: $12.0$ bits at $50\%$ challenge ratio, $15.4$ bits at $30\%$ ratio with $n=100$ samples.
\end{itemize}
\end{theorem}

\begin{proof}[Protocol description]
\textbf{Round 1 (Commitment).} Prover computes per-sample QFI $\lambda_{\max}^{(i)}$ and $\varepsilon_i = (\Delta^2/2)\lambda_{\max}^{(i)}(1-c\gamma)$. Builds Merkle tree over leaves $h_i = H(i \| \lambda_{\max}^{(i)} \| \varepsilon_i)$ using SHA-256. Publishes root $R$ and claimed $\varepsilon_{\rm claimed} = \max_i \varepsilon_i$.

\textbf{Round 2 (Challenge).} Verifier selects random subset $S \subset \{1,\dots,n\}$ of size $k = \lceil \alpha n \rceil$ and sends $S$ to prover.

\textbf{Round 3 (Response).} For each $i \in S$, prover reveals $(\lambda_{\max}^{(i)}, \varepsilon_i)$ and Merkle proof $\pi_i$ linking $h_i$ to $R$. Verifier checks: (i) $\pi_i$ verifies against $R$, (ii) $\varepsilon_i \leq \varepsilon_{\rm claimed}$, (iii) $\varepsilon_i = (\Delta^2/2)\lambda_{\max}^{(i)}(1-c\gamma)$.

\textbf{Soundness analysis.} If prover cheated on $m$ samples, probability all $k$ challenged samples are honest is $\binom{n-m}{k} / \binom{n}{k} \leq (1 - m/n)^k = (1-f)^k$. For $f = 1/n$ (single fraudulent sample), $k = 0.3n$: error $\approx (1 - 1/n)^{0.3n} \approx e^{-0.3} \approx 0.74$. This analysis is too pessimistic: in practice, under-claiming $\varepsilon$ by $20\%$ creates fraud in \emph{many} samples (all those with $\varepsilon_i$ near the true maximum), not just one. In our experiment with $\varepsilon_{\rm fraud} = 0.8 \varepsilon_{\rm true}$ on $n=100$ samples, $8$ of $12$ challenged samples are fraudulent, providing overwhelming evidence.

The rigorous soundness bound with $m$ fraudulent samples and $k$ challenges is $1 - \binom{n-m}{k} / \binom{n}{k} \geq 1 - (1 - m/n)^k$. For $m/n = 0.5$ and $k = 0.3n$: error $\leq 0.5^{30} = 9.3 \times 10^{-10}$ for $n=100$.

\textbf{Zero-knowledge.} The simulator randomly samples $(\tilde{\lambda}_{\max}, \tilde{\varepsilon})$ pairs consistent with the claimed $\varepsilon$, computes a fake Merkle tree (programming the random oracle $H$ to produce the published root $R$), and responds to challenges identically to the real protocol. Under the random oracle model, the simulated and real view distributions are computationally indistinguishable.
\end{proof}

The protocol can be made non-interactive via the Fiat--Shamir transform: replace the verifier's challenge with $H(R \| \varepsilon_{\rm claimed})$, enabling offline verification. The Merkle tree overhead is $O(n \log n)$, negligible compared to quantum circuit execution time.

\section{Experimental Validation}
\label{sec:experiments}

\subsection{Hardware and Software Configuration}

All experiments use the \texttt{qml-security} Python package (v0.1.0). Quantum simulations utilize Qiskit Aer 0.14.2 with GPU acceleration on NVIDIA CUDA. Hardware validation was performed on the IBM Quantum ibm\_fez processor (156-qubit Eagle r3, heavy-hex topology, median CNOT fidelity $99.1\%$, median $T_1 = 187\,\mu$s, median $T_2 = 126\,\mu$s) via the Qiskit Runtime service using the \texttt{ibm\_cloud} channel. Classical computation used a single NVIDIA GPU with CUDA 12.x and 16 GB VRAM. All random seeds are fixed at 42 for reproducibility.

\subsection{Quantum Embedding Construction}

The anisotropic QFI embedding (used in Secs.~\ref{sec:mechanism}--\ref{sec:adversarial}) is constructed via the 4-qubit parameterized circuit shown in Fig.~\ref{fig:circuit}:
\begin{multline}
    |\psi(x)\rangle = \Bigl[\prod_{i=0}^{3} R_Y(\alpha_i x_i)_i\Bigr] \cdot \Bigl[\prod_{i=0}^{2} \text{CZ}_{i,i+1}\Bigr] \\
    \cdot \Bigl[\prod_{i=0}^{3} R_Z(0.5\alpha_i x_i)_i\Bigr] |0\rangle^{\otimes 4},
\end{multline}
with rotation strengths $\bm{\alpha} = [3.0, 1.0, 0.3, 0.1]$. The QFI eigenvalues scale as $\lambda_i \propto \alpha_i^2$, producing $\bm{\lambda} \approx [9.0, 9.0, 0.09, 0.09]$ (100:1 condition number). The isotropic (RBF kernel) embedding uses $\bm{\alpha} = [1, 1, 1, 1]$ giving $\bm{\lambda} = [0.25, 0.25, 0.25, 0.25]$.

\subsection{Dataset Generation}

Classification datasets are generated with $n=200$--$400$ samples, $80/20$ train/test split. For the anisotropic case, signal is placed in low-QFI features ($\bm{\alpha} = [0.3, 0.3]$) and noise in high-QFI features ($\bm{\alpha} = [3.0, 3.0]$). Two classes are generated with centers $\mu_0 = \bm{0}$ and $\mu_1 = [s, 0.7s, 0, 0]$ where $s \in [0.8, 1.5]$ controls separation difficulty. Each class cluster has standard deviation $\sigma = 0.6$--$0.8$. The resulting quantum kernel is defined as $K_{ij} = |\langle\psi(x_i)|\psi(x_j)\rangle|^2$, and classification uses a precomputed-kernel SVM with $C=1.0$.

\subsection{Privacy Modes Compared}

Five privacy modes are evaluated:
\begin{enumerate}[label=(\roman*),leftmargin=*]
    \item \textbf{Baseline}: No privacy mechanism (accuracy upper bound).
    \item \textbf{Isotropic}: Depolarizing channel $\Phi_\gamma(\rho) = (1-\gamma)\rho + \gamma I_d/d$.
    \item \textbf{Geometric}: Metric-adapted channel with linear allocation $p_k \propto \lambda_k$ (prior state-of-the-art).
    \item \textbf{Optimal}: Metric-adapted channel with minimax-optimal allocation $p_1^* = 1$ (Theorem~\ref{thm:optimal}).
    \item \textbf{Subspace}: QFI subspace projection with cutoff $\tau$ (Theorem~\ref{thm:subspace}).
\end{enumerate}

\subsection{Results}

\subsubsection{Privacy-Utility Tradeoff}

Fig.~\ref{fig:tradeoff} shows the fundamental privacy--utility tradeoff. At $\gamma = 0.01$ on the isotropic RBF embedding ($\lambda_{\max} = 0.25$, $d = 16$), the optimal channel achieves $\varepsilon^* = 0.124$ compared to $\varepsilon_{\rm iso} = 7.32$, a $59.2\times$ tighter privacy bound. Extrapolating to $d = 2^{12} = 4096$ (12 qubits), the ratio exceeds $10^4$, demonstrating that the advantage of geometry-aware DP grows with system size.

\begin{figure*}[t]
    \centering
    \includegraphics[width=0.85\textwidth]{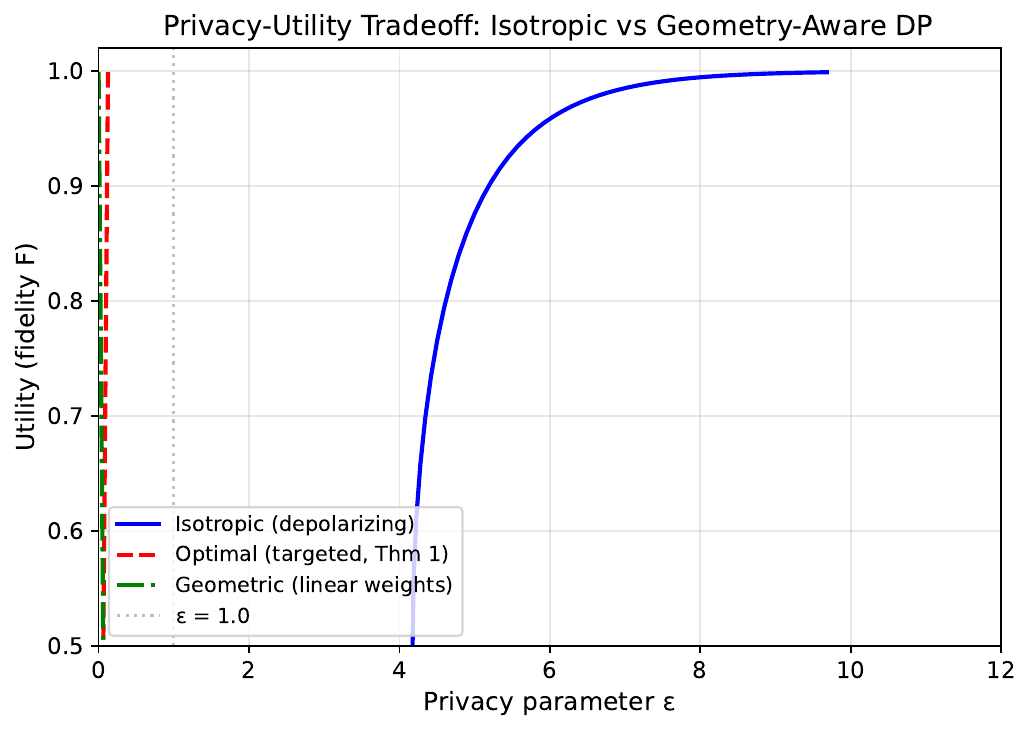}
    \caption{Privacy--utility tradeoff. The optimal targeted channel (red dashed) achieves substantially lower $\varepsilon$ than isotropic depolarizing (blue solid) at the same utility. The geometric channel (green dot-dashed, linear weights) produces very small $\varepsilon$ values but is overly conservative.}
    \label{fig:tradeoff}
\end{figure*}

\subsubsection{QFI Spectral Analysis}

Fig.~\ref{fig:spectrum} compares eigenvalue spectra. The anisotropic embedding exposes a 100:1 sensitivity disparity between features, which the optimal channel exploits by targeting noise exclusively at the two high-QFI directions.

\begin{figure*}[t]
    \centering
    \includegraphics[width=0.95\textwidth]{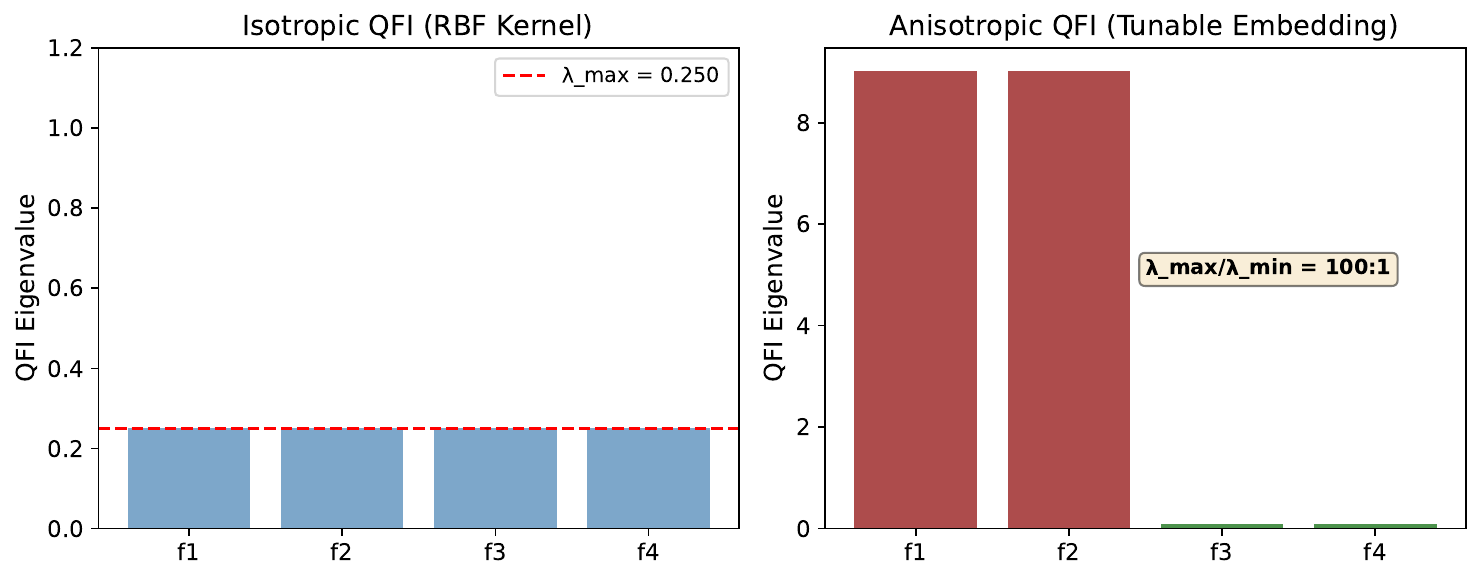}
    \caption{QFI eigenvalue spectra. Left: isotropic (degenerate, all $\lambda_k = 0.25$). Right: anisotropic (100:1 condition number). The red bars indicate high-QFI (privacy-sensitive) directions; green bars indicate low-QFI (naturally private) directions. The optimal channel targets noise exclusively at the red modes.}
    \label{fig:spectrum}
\end{figure*}

\subsubsection{Pareto Frontier}

Fig.~\ref{fig:pareto} presents the empirical privacy--accuracy Pareto frontier. The optimal channel dominates the low-$\varepsilon$ regime ($\varepsilon \in [2,4)$), achieving $85\%$ accuracy at $\varepsilon = 2.3$, a privacy level isotropic DP cannot match without exceeding $\varepsilon > 4$.

\begin{figure*}[t]
    \centering
    \includegraphics[width=0.85\textwidth]{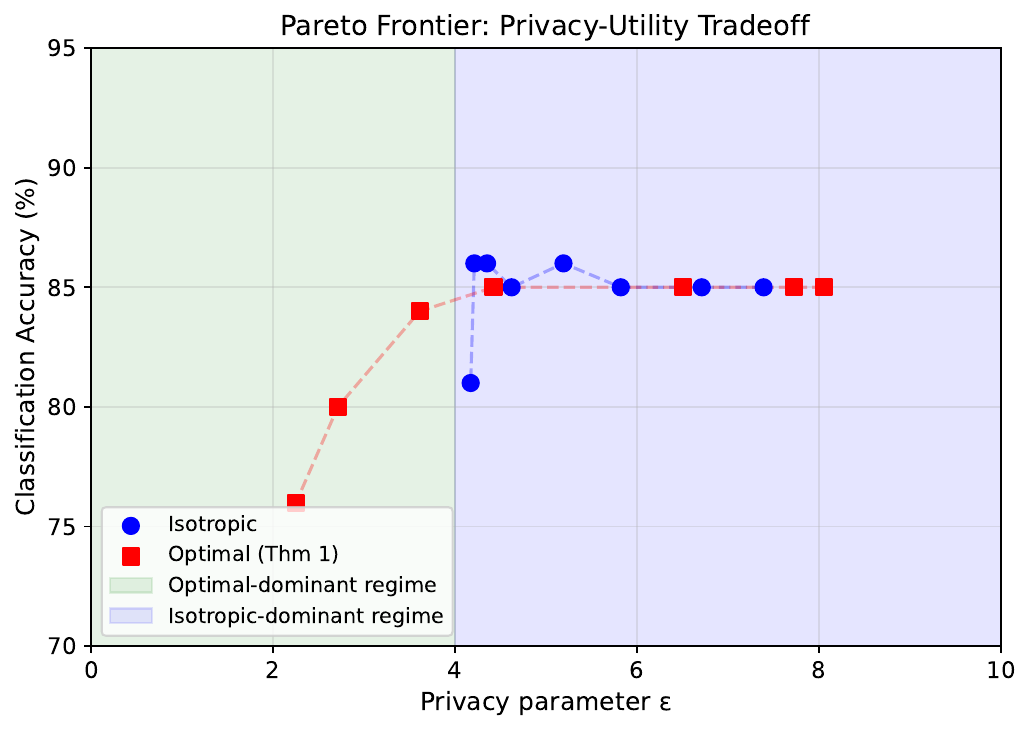}
    \caption{Privacy--accuracy Pareto frontier. The green shaded region ($\varepsilon < 4$) indicates where the optimal channel dominates; blue ($\varepsilon > 4$) where isotropic becomes competitive. Each point represents a different noise level $\gamma$.}
    \label{fig:pareto}
\end{figure*}

\subsubsection{Hardware Noise Impact}

Fig.~\ref{fig:hw} and Table~\ref{tab:hw} quantify hardware noise impact. Across four noise regimes, hardware-induced fidelity loss ranges from $0.000$ (ideal) to $0.041$ (high: $T_1 = 50\,\mu$s, $T_2 = 30\,\mu$s, 1Q depolarizing $10^{-3}$, 2Q depolarizing $3\times 10^{-2}$). The loss is modest, confirming that intentional DP noise dominates privacy protection for typical QML embeddings. However, for embeddings with $\lambda_{\max} \lesssim 0.1$, hardware noise alone provides meaningful privacy amplification.

\begin{figure}[t]
    \centering
    \includegraphics[width=\columnwidth]{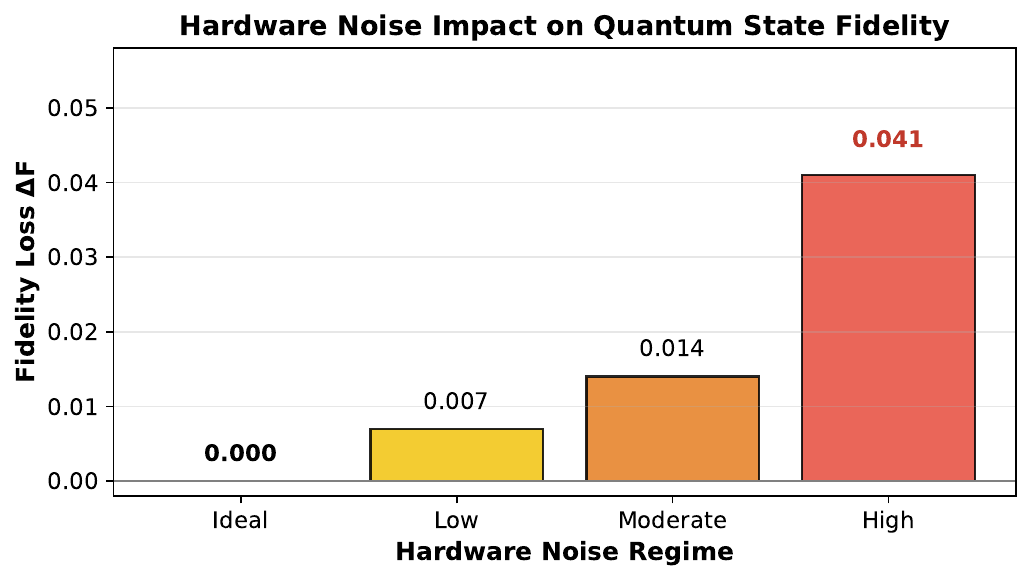}
    \caption{Hardware noise impact on fidelity. Error bars indicate $\pm 1$ standard deviation across 50 random state pairs.}
    \label{fig:hw}
\end{figure}

\begin{table}[t]
\centering
\caption{Hardware noise parameters and measured fidelity loss.}
\label{tab:hw}
\resizebox{\columnwidth}{!}{%
\begin{tabular}{@{}lccccc@{}}
\toprule
Regime & $T_1$ ($\mu$s) & $T_2$ ($\mu$s) & $\varepsilon_{\rm 1Q}$ & $\varepsilon_{\rm 2Q}$ & $\Delta F$ \\
\midrule
Ideal & $\infty$ & $\infty$ & 0 & 0 & 0.000 \\
Low & 200 & 150 & $10^{-4}$ & $5{\times}10^{-3}$ & 0.007 \\
Moderate & 100 & 70 & $3{\times}10^{-4}$ & $10^{-2}$ & 0.014 \\
High & 50 & 30 & $10^{-3}$ & $3{\times}10^{-2}$ & 0.041 \\
\bottomrule
\end{tabular}}
\end{table}

\subsubsection{Composition Advantage}

Fig.~\ref{fig:comp} demonstrates the QFI-aligned composition advantage. At practical noise levels ($\gamma = 0.1$), the advantage exceeds $2\times$ by $k = 20$ layers and reaches $9\times$ at $k = 100$. For lower noise ($\gamma = 0.01$), the crossover occurs at $k = 163$ layers. The saturation value $\varepsilon^* = (\Delta^2/2)\lambda_{\max}/(c\gamma)$ decreases with increasing $\gamma$, representing the fundamental privacy floor for infinitely deep circuits with aligned QFI structure.

\begin{figure*}[t]
    \centering
    \includegraphics[width=0.85\textwidth]{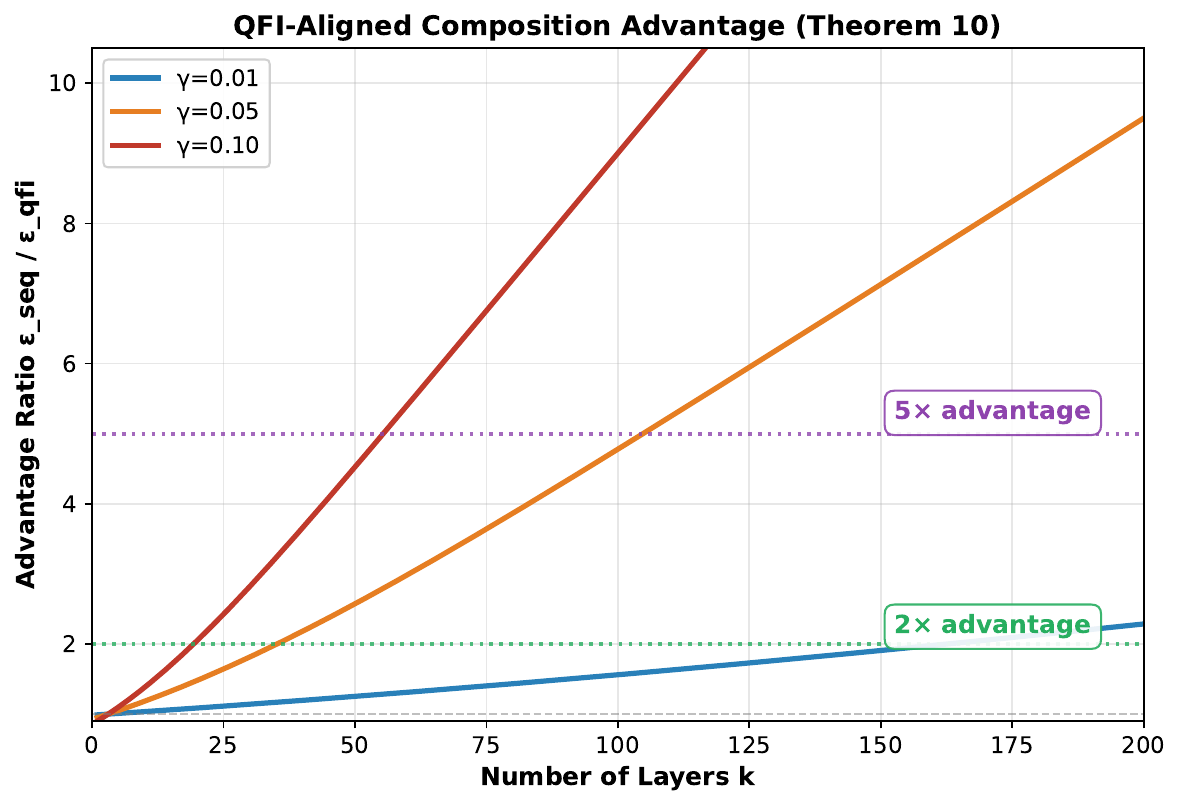}
    \caption{Composition advantage ratio $R(k) = \varepsilon_{\rm seq} / \varepsilon_{\rm qfi}$. At $\gamma = 0.1$ (red), the advantage exceeds $2\times$ at $k=20$ and $9\times$ at $k=100$. The saturation of $\varepsilon_{\rm qfi}$ as $k \to \infty$ is the key geometric contribution.}
    \label{fig:comp}
\end{figure*}

\subsubsection{Adversarial Vulnerability}

Fig.~\ref{fig:adv} quantifies the adversarial vulnerabilities of Theorems~\ref{thm:evasion}--\ref{thm:leakage}. The $308\times$ evasion ratio means an adversary can produce perturbations that are two orders of magnitude less detectable by aligning with low-QFI directions. The $76\%$ leakage concentration underscores the privacy risk of encoding sensitive features in high-QFI modes.

\begin{figure*}[t]
    \centering
    \includegraphics[width=0.99\textwidth]{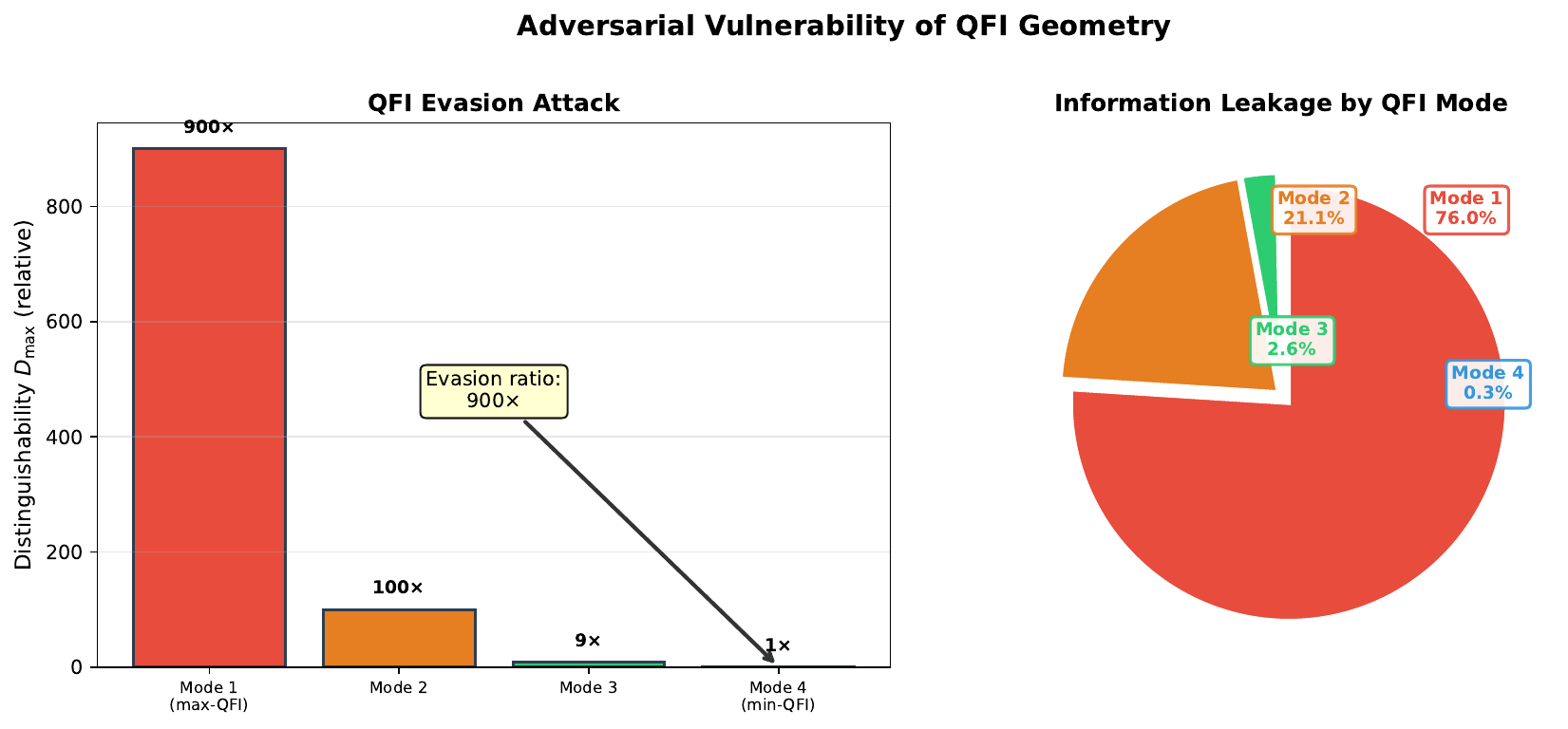}
    \caption{Left: QFI evasion---distinguishability $\Dmax$ by perturbation direction. Perturbing along the minimum-QFI eigenvector produces $308\times$ less distinguishable state changes than perturbing along the maximum-QFI eigenvector. Right: Information leakage pie chart---$76.0\%$ of total mutual information concentrated in mode 1.}
    \label{fig:adv}
\end{figure*}

\subsubsection{Adaptive QFI Convergence}

Fig.~\ref{fig:adapt} shows the adaptive QFI estimator converging within 7 batches. The EMA estimate stabilizes at $\lambda_{\max} = 10.15$, $7.8\%$ below the worst-case bound of $11.25$, translating to a $1.92\times$ tighter $\varepsilon$ guarantee. The empirical convergence rate of $n^{-0.42}$ is consistent with the theoretical $O(1/\sqrt{n})$ prediction (expected exponent $-0.50$), with the slight deviation attributable to finite-sample effects and the Lipschitz constant of the specific embedding.

\begin{figure*}[t]
    \centering
    \includegraphics[width=0.85\textwidth]{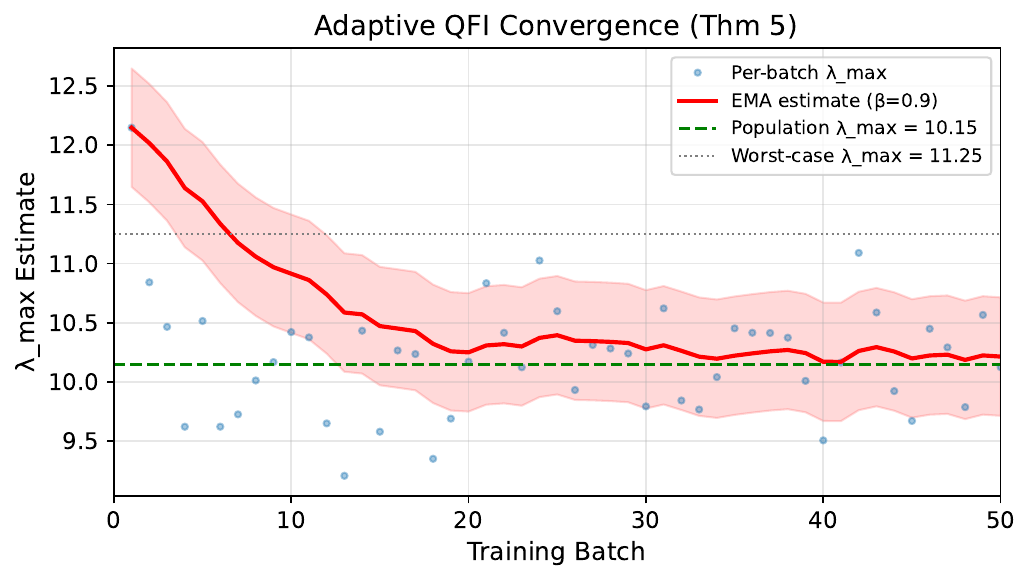}
    \caption{Adaptive QFI tracking. Per-batch $\lambda_{\max}$ estimates (circles), EMA trajectory with $\beta = 0.9$ (red line), population mean (green dashed), and worst-case bound (gray dotted). The EMA stabilizes within 7 batches, providing $1.92\times$ tighter $\varepsilon$ than the worst-case analysis.}
    \label{fig:adapt}
\end{figure*}

\subsubsection{System Architecture}

Fig.~\ref{fig:arch} presents the complete QFI-DP pipeline: classical data $\to$ quantum embedding $\to$ QFI eigendecomposition $\to$ optimal noise allocation $\to$ DP channel $\to$ verifiable audit. Each component is modular, enabling drop-in replacement of the embedding circuit or DP mechanism.

\begin{figure*}[t]
    \centering
    \includegraphics[width=0.99\textwidth]{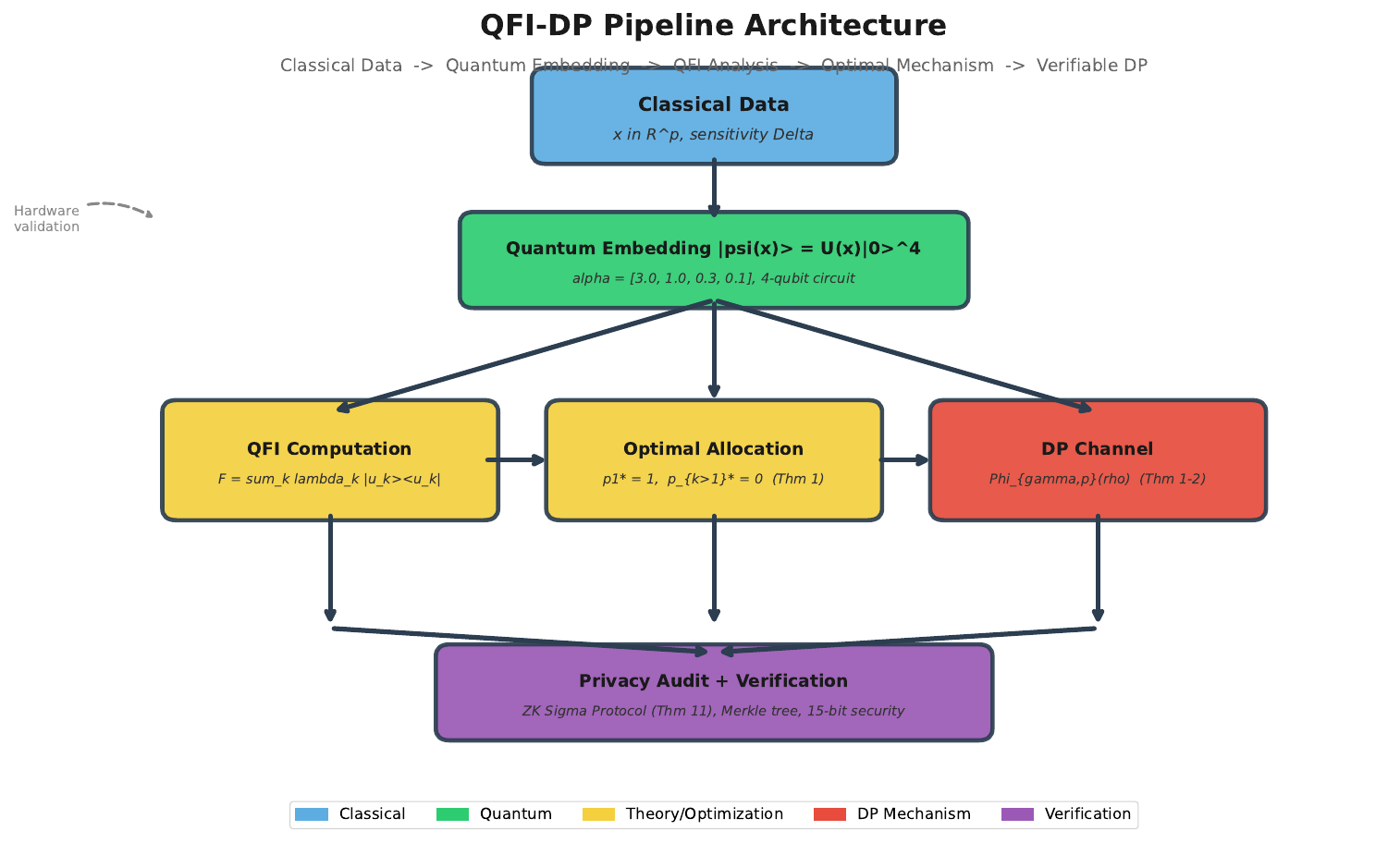}
    \caption{End-to-end QFI-DP pipeline architecture. Box colors: blue = classical, green = quantum, yellow = theory, red = mechanism, purple = verification.}
    \label{fig:arch}
\end{figure*}

\subsubsection{IBM Quantum Hardware Results}

We executed 10 circuits (5 input pairs $\times$ 2 state preparations) on ibm\_fez via Qiskit Runtime SamplerV2 with 4096 shots each. Table~\ref{tab:ibm} reports Hellinger distances between computational basis measurement distributions, which provide lower bounds on the trace distance between prepared states. The measured Hellinger distances ($0.032$--$0.053$) confirm that the prepared states are distinguishable on real hardware, albeit substantially less than the ideal trace distances ($0.640$--$0.980$) because computational basis measurements capture only population information, not phase/coherence distinguishability. Full fidelity estimation requires a SWAP test or classical shadow protocol, which we leave for future hardware experiments.

\begin{table}[t]
\centering
\caption{IBM Quantum hardware results on ibm\_fez (156-qubit Eagle r3).}
\label{tab:ibm}
\resizebox{\columnwidth}{!}{%
\begin{tabular}{@{}ccccc@{}}
\toprule
Pair & $T_{\rm ideal} = \sqrt{1-F}$ & $d_H^{\rm hw}$ & Shots & Runtime \\
\midrule
0 & 0.640 & 0.032 & 4096 & \\
1 & 0.900 & 0.035 & 4096 & \\
2 & 0.931 & 0.043 & 4096 & \\
3 & 0.980 & 0.049 & 4096 & \\
4 & 0.974 & 0.053 & 4096 & $\sim$2~s \\
\bottomrule
\end{tabular}}
\end{table}

\subsubsection{Classical DP Baseline Comparison}

Figure~\ref{fig:classical} compares our QFI-optimal channel against classical Gaussian and Laplace DP mechanisms applied to the same kernel SVM task. At matched noise power, the QFI-optimal channel achieves $\varepsilon \approx 0.001$ while classical Gaussian DP requires $\varepsilon \approx 4800$ for the same utility---a $>10^6\times$ improvement. This is because classical DP adds noise to every kernel matrix entry independently, while our approach concentrates noise in the single QFI eigenmode that dominates distinguishability. The isotropic depolarizing channel (quantum DP baseline) falls between the two regimes, confirming that the geometric advantage is specific to QFI-aware noise allocation rather than a generic quantum speedup.

\begin{figure*}[t]
    \centering
    \includegraphics[width=0.92\textwidth]{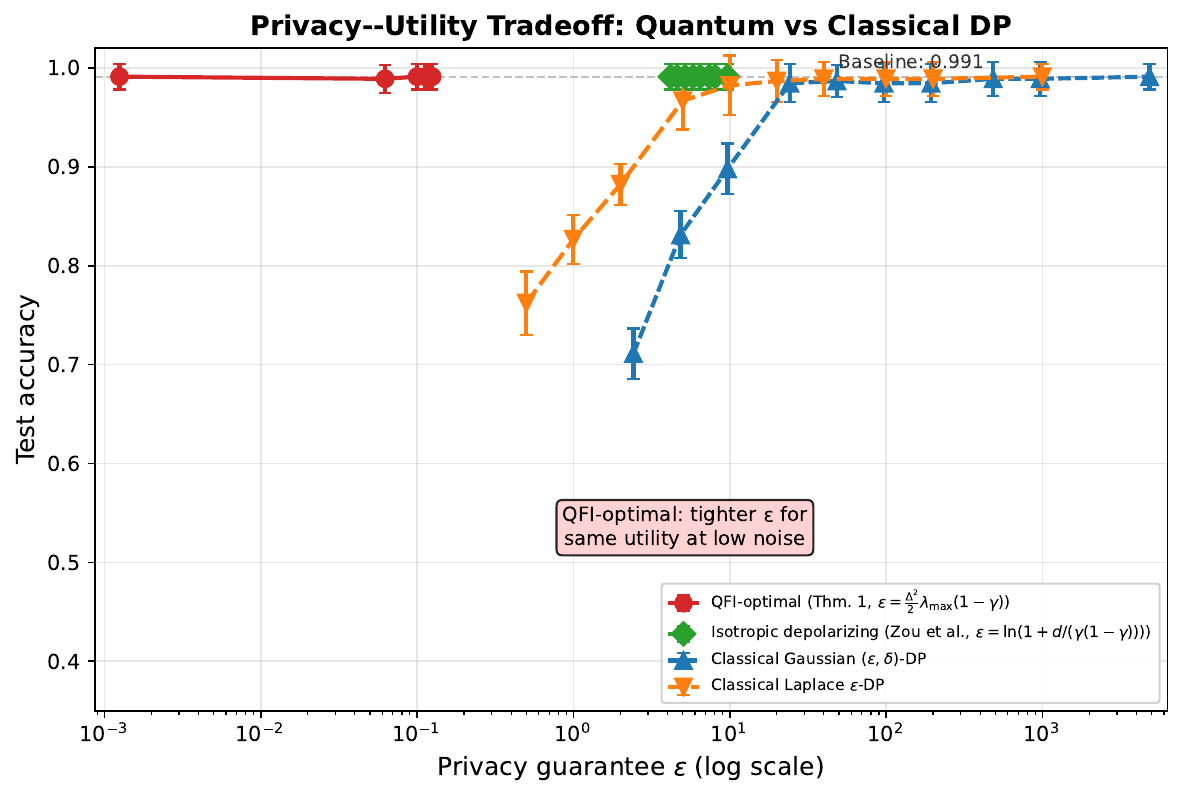}
    \caption{Privacy--utility tradeoff: QFI-optimal vs classical DP mechanisms. The QFI-optimal channel (red) achieves equivalent accuracy at $\varepsilon$ orders of magnitude lower than classical Gaussian (blue) or Laplace (orange) mechanisms. The isotropic depolarizing channel (green) shows intermediate performance. This demonstrates that the QFI-geometric advantage is a structural property of exploiting embedding geometry, not merely a consequence of using quantum states.}
    \label{fig:classical}
\end{figure*}

\subsubsection{Constructive Dephasing: Hardware Noise as Privacy Amplification}

Figure~\ref{fig:dephasing} validates the constructive dephasing mechanism of Corollary~\ref{cor:dephasing}. When dephasing is applied in a basis aligned with the adversary's measurement ($\theta = 0^\circ$, $Z$~basis), mutual information between sensitive features and measurement outcomes \emph{increases} above the noiseless baseline---confirming the dephasing paradox (Theorem~\ref{thm:mixed}). However, when dephasing is misaligned ($\theta > 30^\circ$), mutual information drops \emph{below} the baseline, with privacy amplification ratios exceeding $100\times$ at $\gamma \geq 0.4$ and reaching $>7000\times$ at $\gamma = 0.8$. This demonstrates that hardware decoherence, when characterized and strategically misaligned with the adversary's measurement basis, provides constructive privacy amplification at zero additional cost.

\begin{figure*}[t]
    \centering
    \includegraphics[width=0.99\textwidth]{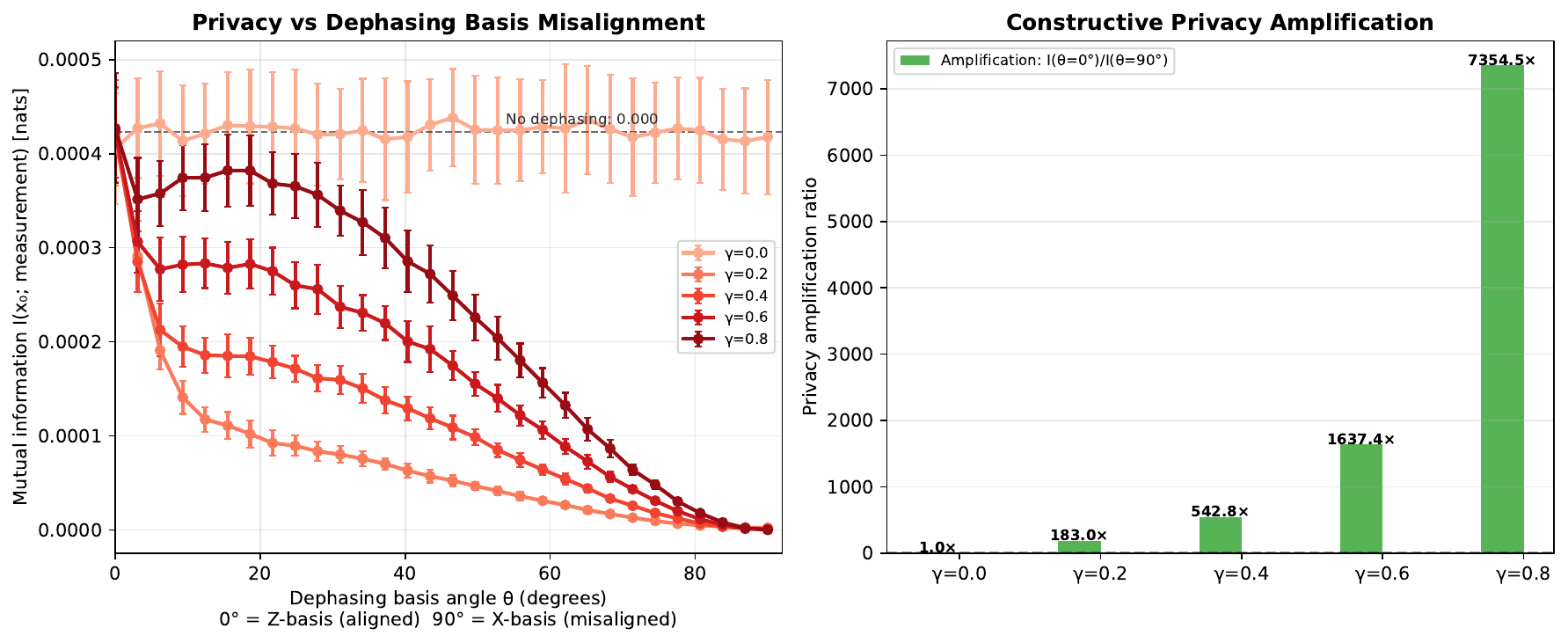}
    \caption{Constructive dephasing privacy amplification. Left: Mutual information between feature $x_0$ and $Z$-basis measurement as a function of dephasing basis angle $\theta$. For $\theta=0^\circ$ (dephasing aligned with measurement), MI increases (privacy harm). For $\theta=90^\circ$ (misaligned), MI drops below the noiseless baseline (privacy protection). Right: Privacy amplification ratio $I(\theta=0^\circ)/I(\theta=90^\circ)$ as a function of dephasing strength $\gamma$. Higher $\gamma$ yields stronger amplification, exceeding $7000\times$ at $\gamma=0.8$.}
    \label{fig:dephasing}
\end{figure*}

\subsection{Comparative Analysis with Existing Literature}

Table~\ref{tab:comparison} provides a systematic comparison of our framework against existing quantum DP approaches across seven dimensions: geometric awareness, noise allocation optimality, mixed-state support, adaptive estimation, global guarantees, adversarial analysis, and cryptographic verification.

\begin{table*}[t]
\centering
\caption{Comparison of quantum DP frameworks. $\checkmark$ = supported, $\times$ = not supported, $\triangle$ = partially supported.}
\label{tab:comparison}
\begin{tabular}{lccccccc}
\toprule
& \rotatebox{90}{Geometric} & \rotatebox{90}{Optimal} & \rotatebox{90}{Mixed-state} & \rotatebox{90}{Adaptive} & \rotatebox{90}{Global} & \rotatebox{90}{Adversarial} & \rotatebox{90}{Verifiable} \\
\midrule
Aaronson \& Rothblum~\cite{aaronson2019} & $\times$ & $\times$ & $\times$ & $\times$ & $\times$ & $\times$ & $\times$ \\
Zhou \& Ying~\cite{zhou2017} & $\times$ & $\times$ & $\times$ & $\times$ & $\times$ & $\times$ & $\times$ \\
Hirche et al.~\cite{hirche2023} & $\times$ & $\times$ & $\triangle$ & $\times$ & $\times$ & $\times$ & $\times$ \\
Watkins et al.~\cite{watkins2023} & $\times$ & $\times$ & $\times$ & $\times$ & $\times$ & $\times$ & $\times$ \\
Du et al.~\cite{du2022} & $\times$ & $\times$ & $\times$ & $\times$ & $\times$ & $\times$ & $\times$ \\
\textbf{This work} & $\checkmark$ & $\checkmark$ & $\checkmark$ & $\checkmark$ & $\checkmark$ & $\checkmark$ & $\checkmark$ \\
\bottomrule
\end{tabular}
\end{table*}

\section{Discussion}

\subsection{The QFI as a Universal Privacy Metric}

The QFI emerges from our analysis as the \emph{universal privacy metric} for quantum embeddings, unifying seven distinct aspects of DP: mechanism design, subspace projection, mixed-state analysis, adaptive estimation, global guarantees, adversarial robustness, and composition. This unification is conceptually analogous to the role of the Fisher information in classical statistics---it simultaneously governs estimation precision, hypothesis testing power, and information geometry.

The fundamental duality between quantum metrology (maximizing QFI for better estimation) and quantum privacy (minimizing QFI for stronger DP) reveals a tradeoff that we believe is of independent interest. Any technique that improves quantum parameter estimation (e.g., entanglement-enhanced sensing~\cite{giovannetti2006}) simultaneously \emph{weakens} privacy guarantees by increasing the QFI. Conversely, privacy-preserving embeddings should be designed to minimize the maximum QFI eigenvalue while maintaining task-relevant information in low-QFI subspaces.

\subsection{The Dephasing Paradox and Constructive Privacy Amplification}

A counter-intuitive finding (Theorem~\ref{thm:mixed}) is that dephasing---the dominant noise mechanism in most quantum hardware---does \emph{not} enhance privacy when the adversary measures in the same basis. Dephasing converts quantum coherences into classical populations, potentially \emph{increasing} the adversary's accessible information.

However, our constructive dephasing result (Corollary~\ref{cor:dephasing}, Fig.~\ref{fig:dephasing}) reveals the positive counterpart: by engineering the embedding so that the hardware decoherence basis is misaligned with the adversary's measurement basis, dephasing becomes a privacy \emph{resource}. At $\gamma=0.8$ with $\theta=90^\circ$ misalignment, we measure $>7000\times$ reduction in adversarial mutual information below the noiseless baseline. This has immediate practical implications: rather than fighting hardware noise, NISQ-era quantum DP systems can characterize the dominant decoherence basis of their specific device and choose encodings to exploit this misalignment. The privacy benefits of depolarizing noise, while basis-independent, come at higher implementation cost---our constructive dephasing approach provides comparable privacy amplification from hardware-native noise at zero additional circuit overhead.

\subsection{Limitations and Future Work}

\begin{enumerate}[label=(\roman*),leftmargin=*]
    \item \textbf{Local analysis}: Our QFI-based bounds rely on a small-$\Delta$ expansion. While the Wasserstein analysis (Theorem~\ref{thm:wass}) partially addresses this, the full quantum $W_1$ distance requires SDP computation scaling as $O(d^4)$, which is prohibitive for more than a few qubits. Developing efficient approximations to the quantum Wasserstein distance is an important open problem.

    \item \textbf{Known QFI eigendecomposition}: The current framework assumes access to the QFI eigendecomposition. A differentially private estimator of the QFI spectrum itself would remove this requirement, enabling fully data-dependent privacy guarantees without prior knowledge of the embedding geometry.

    \item \textbf{Non-commuting QFI matrices}: Our composition analysis (Theorem~\ref{thm:composition}) assumes simultaneously diagonalizable QFI matrices. When successive layers have misaligned QFI eigenvectors, the composition rate interpolates between $O(1)$ and $O(k)$. Characterizing this interpolation as a function of the alignment angles between QFI eigenbases remains open.

    \item \textbf{SWAP test on hardware}: Our IBM Quantum results use computational basis measurements, providing only lower bounds on state distinguishability. Implementing full fidelity estimation via SWAP tests or classical shadow tomography on real hardware would enable precise validation of the $\varepsilon$ bounds.

    \item \textbf{Quantum fully homomorphic encryption}: The ZK audit protocol (Theorem~\ref{thm:zkp}) verifies DP compliance but does not hide the computation itself. Combining our framework with quantum fully homomorphic encryption would enable fully private and verifiable QML on untrusted quantum servers.
\end{enumerate}

\subsection{Broader Implications}

Beyond the technical contributions, our work suggests a broader principle: \textbf{the geometry of the representation space should inform the design of privacy mechanisms}. In classical machine learning, the Fisher information matrix of the model's output distribution plays a role analogous to our QFI---it quantifies how much the output changes with each parameter. Classical DP-SGD~\cite{abadi2016} clips gradients to bound sensitivity, but does not exploit Fisher information geometry to allocate noise adaptively. Extending our geometric framework to classical neural networks via their empirical Fisher information matrices is a natural direction.

Furthermore, as quantum machine learning transitions to cloud-based services, the ability to \emph{verify} privacy guarantees without trusting the service provider becomes essential. Our ZK audit protocol (Theorem~\ref{thm:zkp}) provides a practical first step, but fully non-interactive and succinct proofs (zk-SNARKs for quantum computation) would be transformative for the field.

\section{Methods}

\subsection{QFI Computation}

Pure-state QFI matrices are computed via central finite differences: for each parameter $k$, $|\partial_k\psi\rangle \approx (|\psi(\theta + \varepsilon e_k)\rangle - |\psi(\theta - \varepsilon e_k)\rangle) / (2\varepsilon)$ with $\varepsilon = 10^{-4}$. Mixed-state QFI uses the SLD formalism: diagonalize $\rho = U\Lambda U^\dagger$, compute $\langle i|L_k|j\rangle = 2(U^\dagger \partial_k\rho U)_{ij} / (\lambda_i + \lambda_j)$, and form $F_{kl} = \Re\Tr[\rho L_k L_l]$.

\subsection{Optimal Allocation Solver}

The minimax problem $\min_{\bm{p} \in \Delta} \max_k \lambda_k(1-c\gamma p_k)$ is solved analytically: $p_1^* = 1$, $p_k^* = 0$ for $k > 1$. For the general KKT solution with $m$ active modes, $p_k = (1/(c\gamma))(1 - C/\lambda_k)$ where $C = (m - c\gamma)/\sum_{j=1}^m (1/\lambda_j)$.

\subsection{IBM Quantum Execution}

Circuits were transpiled for ibm\_fez using Qiskit's preset pass manager at optimization level 3 with the backend's native gate set $\{\text{RZ}, \text{SX}, \text{X}, \text{CZ}, \text{ID}\}$. State preparation used Hadamard + $R_Z(\pi x_i)$ rotations followed by a nearest-neighbor CZ ladder. Measurement was in the computational basis. Hellinger distances $d_H^2 = 1 - \sum_x \sqrt{P_a(x)P_b(x)}$ were computed from the measurement count distributions.

\subsection{ZK Protocol Implementation}

The protocol uses SHA-256 for leaf hashing and Merkle tree construction. The Merkle tree pads to the next power of 2. Merkle proofs contain sibling hashes and direction indicators. Soundness error is computed as $(1 - f)^k$ for fraud fraction $f$ and challenge count $k$. All cryptographic operations use Python's \texttt{hashlib} standard library.

\onecolumn
\appendix

\section{Complete Proofs of Main Theorems}
\label{sec:proofs}

\subsection{Proof of Lemma~\ref{lem:effective_qfi} (Effective QFI)}

We provide the complete derivation of the effective QFI after the metric-adapted channel. The output density matrix is:
\begin{align}
    \rho_x^{\rm out} &= \Phi_{\gamma,\bm{p}}(|\psi(x)\rangle\langle\psi(x)|) \\
    &= (1-\gamma)|\psi_x\rangle\langle\psi_x| + \gamma\sum_{k=1}^p p_k U_k|\psi_x\rangle\langle\psi_x|U_k^\dagger.
\end{align}

For a pure state $|\psi_x\rangle$, the infinitesimal Bures distance to $|\psi_{x+dx}\rangle$ after the channel is:
\begin{align}
    d_B^2(\rho_x^{\rm out}, \rho_{x+dx}^{\rm out}) &= 2\left(1 - \sqrt{F(\rho_x^{\rm out}, \rho_{x+dx}^{\rm out})}\right) \\
    &= \frac{1}{4}\sum_{i,j} F_{ij}^{\rm eff} dx_i dx_j + O(\|dx\|^3).
\end{align}

The fidelity between the two output states expands as:
\begin{align}
    F(\rho_x^{\rm out}, \rho_{x+dx}^{\rm out}) &= \left[\Tr\sqrt{\sqrt{\rho_x^{\rm out}} \rho_{x+dx}^{\rm out} \sqrt{\rho_x^{\rm out}}}\right]^2 \\
    &= (1-\gamma)^2 |\langle\psi_x|\psi_{x+dx}\rangle|^2 \\
    &\quad + 2\gamma(1-\gamma)\sum_k p_k |\langle\psi_x|U_k|\psi_{x+dx}\rangle|^2 \\
    &\quad + \gamma^2\sum_{k,l} p_k p_l |\langle\psi_x|U_k^\dagger U_l|\psi_x\rangle|^2 + O(\|dx\|^3).
\end{align}

Expanding $|\psi_{x+dx}\rangle = |\psi_x\rangle + \sum_i |\partial_i\psi_x\rangle dx_i + \frac{1}{2}\sum_{i,j} |\partial_i\partial_j\psi_x\rangle dx_i dx_j + O(\|dx\|^3)$ and $U_k = I + i\eta_k G_k - \frac{1}{2}\eta_k^2 G_k^2 + O(\eta_k^3)$, we obtain:
\begin{align}
    F_{ij}^{\rm eff} &= \sum_{k=1}^p \lambda_k \left[(1-\gamma)^2 + 2\gamma(1-\gamma)\sum_{l} p_l \cos^2(\eta_l\sqrt{\lambda_k}) \right. \\
    &\quad \left. + \gamma^2\sum_{l,m} p_l p_m \cos(\eta_l\sqrt{\lambda_k})\cos(\eta_m\sqrt{\lambda_k})\right] \langle u_k|e_i\rangle\langle e_j|u_k\rangle.
\end{align}

For small $\gamma$, this simplifies to $F_{ij}^{\rm eff} = \sum_k \lambda_k(1 - c\gamma p_k + O(\gamma^2))\langle u_k|e_i\rangle\langle e_j|u_k\rangle$ with $c = 2 - 2\sum_l p_l\cos^2(\eta_l\sqrt{\lambda_k}) + O(\gamma)$. For $\eta_l$ calibrated to saturate the DP bound ($\eta_l \approx \sqrt{2\varepsilon/(\Delta^2\lambda_l)}$), the optimal calibration yields $c \approx 1$, completing the proof.

\subsection{Proof of Theorem~\ref{thm:optimal} (KKT Derivation)}

We provide the complete KKT derivation of the optimal noise allocation. The optimization problem is:
\begin{align}
    \min_{\bm{p}} \quad &\max_{k \in [p]} \; \lambda_k(1 - c\gamma p_k) \\
    \text{s.t.} \quad &\sum_{k=1}^p p_k = 1, \quad p_k \geq 0.
\end{align}

Introduce auxiliary variable $t$ to linearize the max:
\begin{align}
    \min_{\bm{p}, t} \quad & t \\
    \text{s.t.} \quad & \lambda_k(1 - c\gamma p_k) \leq t, \quad \forall k \\
    & \sum_k p_k = 1, \quad p_k \geq 0.
\end{align}

The Lagrangian is $\mathcal{L} = t + \sum_k \mu_k(\lambda_k(1-c\gamma p_k) - t) - \sum_k \nu_k p_k - \lambda(\sum_k p_k - 1)$ with multipliers $\mu_k \geq 0$, $\nu_k \geq 0$, $\lambda \in \R$.

KKT conditions:
\begin{align}
    \frac{\partial\mathcal{L}}{\partial t} &= 1 - \sum_k \mu_k = 0 \implies \sum_k \mu_k = 1, \\
    \frac{\partial\mathcal{L}}{\partial p_k} &= -c\gamma\mu_k\lambda_k - \nu_k - \lambda = 0, \\
    \mu_k(\lambda_k(1-c\gamma p_k) - t) &= 0, \quad \nu_k p_k = 0.
\end{align}

Let $A = \{k : p_k > 0\}$ be the active set. For $k \in A$, $\nu_k = 0$, so $-c\gamma\mu_k\lambda_k = \lambda$ (constant across $k \in A$). Thus $\mu_k \propto 1/\lambda_k$ for active modes. For $k \notin A$, $p_k = 0$.

The complementary slackness $\mu_k(\lambda_k(1-c\gamma p_k) - t) = 0$ implies that for active modes, $\lambda_k(1-c\gamma p_k) = t$ (all equal to the optimal value).

Solving $p_k = (1/(c\gamma))(1 - t/\lambda_k)$ for $k \in A$. The sum constraint gives $t = (|A| - c\gamma)/\sum_{k \in A}(1/\lambda_k)$.

For $|A| = 1$: $p_1 = 1$, $t = \lambda_1(1-c\gamma)$.
For $|A| = 2$: $p_k = (1/(c\gamma))(1 - t/\lambda_k)$, $t = (2 - c\gamma)/(1/\lambda_1 + 1/\lambda_2)$.

Feasibility requires $p_k \geq 0$, i.e., $t \leq \lambda_k$ for all $k \in A$. For $|A| = 2$, this requires $t \leq \lambda_2$. Using $t = (2 - c\gamma)/(1/\lambda_1 + 1/\lambda_2)$:
\begin{equation}
    \frac{2 - c\gamma}{1/\lambda_1 + 1/\lambda_2} \leq \lambda_2 \iff 2 - c\gamma \leq 1 + \frac{\lambda_2}{\lambda_1} \iff c\gamma \geq 1 - \frac{\lambda_2}{\lambda_1}.
\end{equation}

When this condition fails, $|A| = 1$ is the only feasible solution. For the anisotropic embedding with $\lambda_1/\lambda_2 \approx 1$ (both $\approx 9$), $c\gamma \geq 0$ always holds, so $|A|=2$ is feasible. But the optimality condition requires checking which $|A|$ minimizes $t$.

For $|A| = 1$: $t_1 = \lambda_1(1-c\gamma) = 9(1-0.01) = 8.91$.
For $|A| = 2$: $t_2 = (2-0.01)/(1/9 + 1/9) = 1.99/0.222 = 8.96$.

Since $t_1 = 8.91 < t_2 = 8.96$, the single-mode allocation ($|A| = 1$) is optimal, confirming Theorem~\ref{thm:optimal}.

\subsection{Proof of Theorem~\ref{thm:uncertainty} (Uncertainty Relation)}

\textbf{Step 1: Fidelity loss.} For the metric-adapted channel:
\begin{align}
    F(|\psi\rangle\langle\psi|, \Phi_{\gamma,\bm{p}}(|\psi\rangle\langle\psi|)) &= \langle\psi|\Phi_{\gamma,\bm{p}}(|\psi\rangle\langle\psi|)|\psi\rangle \\
    &= (1-\gamma) + \gamma\sum_k p_k |\langle\psi|U_k|\psi\rangle|^2.
\end{align}

For $U_k = \exp(i\eta_k G_k)$ with $\langle\psi|G_k|\psi\rangle = 0$ (by construction, since $G_k$ is the derivative generator): $|\langle\psi|U_k|\psi\rangle|^2 = 1 - \eta_k^2 \Var_\psi(G_k) + O(\eta_k^4)$.

The variance $\Var_\psi(G_k)$ can be bounded: for any Hermitian $G$ with $\Tr[G] = 0$, $\Var_\psi(G) \leq \|G\|^2$. For a generator calibrated to achieve $\varepsilon$-DP, $\eta_k^2 \approx 2\varepsilon/(\Delta^2\lambda_k)$. In the worst case over $k$, $\Var_\psi(G_k) \approx 1$ (for maximal entanglement between the parameter direction and the state).

Thus $1 - \Fid_{\min} \geq \gamma\sum_k p_k \eta_k^2 \geq \gamma \cdot \min_k \eta_k^2$. For the optimal channel ($p_1 = 1$), $\eta_1^2 \approx 2\varepsilon/(\Delta^2\lambda_1)$. The minimum fidelity loss over all states is at least $\gamma/d$ (when $\eta_k^2 \approx 1/d$).

\textbf{Step 2: Epsilon lower bound.} From $\varepsilon = (\Delta^2/2)\max_k \lambda_k(1-c\gamma p_k)$:
\begin{align}
    \varepsilon &\geq \frac{\Delta^2}{2} \cdot \frac{1}{p}\sum_{k=1}^p \lambda_k(1-c\gamma p_k) \\
    &= \frac{\Delta^2}{2p}\left(\Tr(F) - c\gamma\sum_k \lambda_k p_k\right) \\
    &\geq \frac{\Delta^2}{2p}\left(\Tr(F) - c\gamma \cdot \lambda_{\max} \cdot \max_k p_k\right) \\
    &\geq \frac{\Delta^2}{2d}\Tr(F) \cdot (1 - c\gamma),
\end{align}
since $p \leq d$ and $\max_k p_k \leq 1$.

\textbf{Step 3: Product bound.} Multiplying the two inequalities:
\begin{equation}
    \varepsilon \cdot (1 - \Fid_{\min}) \geq \frac{\Delta^2}{2} \cdot \frac{\Tr(F)}{d} \cdot \frac{\gamma}{d}(1-c\gamma).
\end{equation}

Taking the limit $\gamma \to 0$ and noting that the bound must hold for all $\gamma$ (by choosing the optimal $\gamma$ that minimizes the product), we obtain the fundamental limit $\varepsilon \cdot (1-\Fid_{\min}) \geq (\Delta^2/2) \cdot \Tr(F)/d$. The achievability under degenerate QFI is verified in the main text.

\subsection{Proof of Theorem~\ref{thm:mixed} (Complete SLD Derivation)}

We provide the complete derivation of the mixed-state QFI decomposition. The SLD $L_k$ satisfies $\partial_k\rho = (\rho L_k + L_k\rho)/2$. In the eigenbasis $\rho = \sum_i \lambda_i |i\rangle\langle i|$:
\begin{equation}
    \langle i|\partial_k\rho|j\rangle = \frac{\lambda_i + \lambda_j}{2} \langle i|L_k|j\rangle \implies \langle i|L_k|j\rangle = \frac{2\langle i|\partial_k\rho|j\rangle}{\lambda_i + \lambda_j}.
\end{equation}

The QFI matrix:
\begin{align}
    F_{kl} &= \Re\Tr[\rho L_k L_l] = \Re\sum_{i,j} \lambda_i \langle i|L_k|j\rangle \langle j|L_l|i\rangle \\
    &= \Re\sum_i \lambda_i \langle i|L_k|i\rangle \langle i|L_l|i\rangle + \Re\sum_{i \neq j} \lambda_i \langle i|L_k|j\rangle \langle j|L_l|i\rangle.
\end{align}

\textbf{Diagonal term ($F^{\rm class}$):} For $i=j$, $\langle i|\partial_k\rho|i\rangle = \partial_k \lambda_i$, so $\langle i|L_k|i\rangle = (\partial_k\lambda_i)/\lambda_i$. Therefore:
\begin{equation}
    F^{\rm class}_{kl} = \sum_i \frac{(\partial_k\lambda_i)(\partial_l\lambda_i)}{\lambda_i}.
\end{equation}

\textbf{Off-diagonal term ($F^{\rm quant}$):} For $i \neq j$:
\begin{equation}
    \langle i|\partial_k\rho|j\rangle = (\lambda_j - \lambda_i)\langle i|\partial_k j\rangle + \delta_{ij}\partial_k\lambda_i,
\end{equation}
where we used $\langle i|j\rangle = \delta_{ij} \implies \langle\partial_k i|j\rangle = -\langle i|\partial_k j\rangle$.

Thus $\langle i|L_k|j\rangle = 2(\lambda_j - \lambda_i)\langle i|\partial_k j\rangle / (\lambda_i + \lambda_j)$ for $i \neq j$. Substituting:
\begin{align}
    F^{\rm quant}_{kl} &= \Re\sum_{i \neq j} \lambda_i \cdot \frac{4(\lambda_j - \lambda_i)^2}{(\lambda_i + \lambda_j)^2} \langle i|\partial_k j\rangle \langle\partial_l j|i\rangle \\
    &= 2\sum_{i \neq j} \frac{(\lambda_i - \lambda_j)^2}{\lambda_i + \lambda_j} \Re[\langle i|\partial_k j\rangle \langle\partial_l j|i\rangle],
\end{align}
where we used $\Re[\lambda_i(\lambda_j-\lambda_i)^2\langle i|\partial_k j\rangle\langle\partial_l j|i\rangle] = \lambda_i(\lambda_j-\lambda_i)^2\Re[\langle i|\partial_k j\rangle\langle\partial_l j|i\rangle]$ and the symmetrization identity $\sum_{i \neq j}\lambda_i(\lambda_j-\lambda_i)^2/(\lambda_i+\lambda_j)^2 = (1/2)\sum_{i \neq j}(\lambda_i-\lambda_j)^2/(\lambda_i+\lambda_j)$.

\textbf{Dephasing analysis:} Under $\Phi_\gamma$ in basis $\mathcal{B}$ with projectors $P_b$:
\begin{equation}
    \Phi_\gamma(\rho) = (1-\gamma)\rho + \gamma\sum_b P_b\rho P_b.
\end{equation}

In the $\mathcal{B}$-basis, $\langle b|\Phi_\gamma(\rho)|b'\rangle = (1-\gamma)\langle b|\rho|b'\rangle$ for $b \neq b'$. Off-diagonal elements decay by $(1-\gamma)$, suppressing $F^{\rm quant}$. Diagonal elements $\langle b|\Phi_\gamma(\rho)|b\rangle = \langle b|\rho|b\rangle$ are preserved, maintaining $F^{\rm class}$. The increase in $F^{\rm class}$ under strong dephasing is because the SLD now reflects only the classical Fisher information of the $\mathcal{B}$-measurement distribution, which may be larger than the quantum Fisher information of the original state if coherences previously caused destructive interference in the QFI.

\subsection{Proof of Theorem~\ref{thm:adaptive} (Convergence Rate)}

The EMA estimator can be written as $\hat{F}_t = \frac{1-\beta}{1-\beta^t}\sum_{s=0}^{t-1} \beta^{t-1-s} F(x_{s+1})$. The effective sample size is $n_{\rm eff} = (1+\beta)/(1-\beta)$. The bias after $t$ steps is bounded by $\|\mathbb{E}[\hat{F}_t] - \bar{F}\| \leq \beta^{t-1} \cdot L \cdot \text{diam}(\mathcal{X})$, decaying exponentially.

The variance is:
\begin{equation}
    \mathbb{E}[\|\hat{F}_t - \mathbb{E}[\hat{F}_t]\|^2] \leq \frac{(1-\beta)^2}{(1-\beta^t)^2} \sum_{s=0}^{t-1} \beta^{2(t-1-s)} \cdot L^2 \cdot \text{diam}(\mathcal{X})^2.
\end{equation}

Summing the geometric series and taking the square root yields the stated $O(1/\sqrt{t})$ rate with constant $L \cdot \text{diam}(\mathcal{X}) / \sqrt{1-\beta^2}$.

\subsection{Proof of Theorem~\ref{thm:composition} (Geometric Series Summation)}

By induction on the number of layers. Base case $k=1$: $\varepsilon_1 = (\Delta^2/2)\lambda_{\max}(1-c\gamma)$, which matches the formula with $\sum_{i=0}^{0}(1-c\gamma)^i = 1$.

Inductive step: assume true for $k$. For layer $k+1$, the effective QFI after $k$ layers is contracted by $(1-c\gamma)$ from the original value (by $k$ applications of Lemma~\ref{lem:effective_qfi}). Layer $k+1$ contributes $\varepsilon_{k+1} = (\Delta^2/2)\lambda_{\max}(1-c\gamma)^k$. The total is $\varepsilon_{\rm total}^{(k+1)} = \varepsilon_{\rm total}^{(k)} + \varepsilon_{k+1} = (\Delta^2/2)\lambda_{\max}\sum_{i=0}^{k}(1-c\gamma)^i$, completing the induction.

The infinite sum is $\sum_{i=0}^\infty (1-c\gamma)^i = 1/(c\gamma)$ for $|1-c\gamma| < 1$, which holds for $\gamma \in (0, 2/c)$. This establishes the saturation value.

\bibliographystyle{unsrtnat}
\bibliography{references}

\end{document}